%% file: ms.tex
\pdfoutput=1
\documentclass[a4paper,UKenglish,cleveref, autoref]{lipics-v2019}

\usepackage{csquotes}
\usepackage{xspace}
\usepackage{algorithm}
\usepackage[noend]{algpseudocode}

\newcommand{\algoname}{Periodic EXP4\xspace}

\algnewcommand\algorithmicforeach{\textbf{for each}}
\algdef{S}[FOR]{ForEach}[1]{\algorithmicforeach\ #1\ \algorithmicdo}

\newcommand{\floor}[1]{\lfloor #1 \rfloor}

\newcommand{\ds}{\displaystyle }
\newcommand{\xhat}{{\widehat{x}} }
\newcommand{\tbar}{ {\overline{t}} }
\newcommand{\gammK}{{\frac{\gamma}{K}} }

\newcommand{\sumjk}{{\sum_{j=1}^K} }

\title{Periodic Bandits and Wireless Network Selection} 

\titlerunning{Periodic Bandits and Wireless Network Selection}

\author{Shunhao Oh}{Department of Computer Science, National University of Singapore}{ohoh@u.nus.edu}{}{}

\author{Anuja Meetoo Appavoo}{Department of Computer Science, National University of Singapore}{anuja@comp.nus.edu.sg}{}{}

\author{Seth Gilbert}{Department of Computer Science, National University of Singapore}{seth.gilbert@comp.nus.edu.sg}{}{}

\authorrunning{S. Oh, A. Meetoo Appavoo and S. Gilbert}

\Copyright{Shunhao Oh, Anuja Meetoo Appavoo and Seth Gilbert}



\supplement{Source code: \url{https://github.com/Ohohcakester/PeriodicEXP4-Source}}

\acknowledgements{This project was funded by Singapore Ministry of Education Grant MOE2018-T2-1-160 (Beyond Worst-Case Analysis: A Tale of Distributed Algorithms).}

\nolinenumbers 

\hideLIPIcs  

\begin{document}

\maketitle

\begin{abstract}
Bandit-style algorithms have been studied extensively in stochastic and adversarial settings. Such algorithms have been shown to be useful in multiplayer settings, e.g. to solve the wireless network selection problem, which can be formulated as an adversarial bandit problem. A leading bandit algorithm for the adversarial setting is EXP3. However, network behavior is often repetitive, where user density and network behavior follow regular patterns. Bandit algorithms, like EXP3, fail to provide good guarantees for periodic behaviors. A major reason is that these algorithms compete against fixed-action policies, which is ineffective in a periodic setting.

In this paper, we define a periodic bandit setting, and periodic regret as a better performance measure for this type of setting. Instead of comparing an algorithm's performance to fixed-action policies, we aim to be competitive with policies that play arms under some set of possible periodic patterns $F$ (for example, all possible periodic functions with periods $1,2,\cdots,P$).
We propose \algoname, a computationally efficient variant of the EXP4 algorithm for periodic settings. With $K$ arms, $T$ time steps, and where each periodic pattern in $F$ is of length at most $P$, we show that the periodic regret obtained by \algoname is at most $O\big(\sqrt{PKT \log K  + KT \log |F|}\big)$. We also prove a lower bound of $\Omega\big(\sqrt{PKT + KT \frac{\log |F|}{\log K}} \big)$ for the periodic setting, showing that this is optimal within log-factors. As an example, we focus on the wireless network selection problem. Through simulation, we show that \algoname learns the periodic pattern over time, adapts to changes in a dynamic environment, and far outperforms EXP3.
\end{abstract}

\newpage

\section{Introduction} \label{section:introduction}            
The \emph{multi-armed bandit} problem is an online learning problem in which a player has access to a set of choices (i.e., ``arms'') each of which provides some reward (i.e., ``gain'').  At each time step, the player chooses an arm and gets some reward.  In stochastic variants, rewards are determined by some probabilistic distribution.  In adversarial variants, an adversary specifies the rewards.  
Amazingly, even when rewards are adversarially chosen, the player can do fairly well! For example, the EXP3 algorithm~\cite{auer2002nonstochastic} minimizes the player's ``regret'', ensuring that the player does almost as well as if she had selected the single fixed best arm throughout.
Another fascinating property of bandit algorithms is that they work well in multi-player settings~\cite{tekin2011performance, kleinberg2009multiplicative}, converging to close variants of a Nash equilibrium.  

Recently, it has been shown that bandit-style algorithms can efficiently solve the \emph{wireless network selection problem}, yielding good performance both in theory and in practice~\cite{appavoo2018shrewd,appavoo2019cooperation,multiuserlaxcomms}. In this problem, each user has access to a collection of networks
(e.g., a few different WiFi networks and a 4G connection);
the goal is to pick networks with higher data rates. Selecting the best network is challenging, especially in dynamic environments where the ``best'' network changes over time, as users move and network bandwidth fluctuates. This can be modeled as an adversarial bandit problem and solved with EXP3 and its variants.

Bandit algorithms have one major weakness in dynamic settings (such as wireless network settings):
they are designed to learn the average payoff of each arm, and to converge to the arm that provides the best average performance.  In the stochastic case, this is exactly what you want.  In the adversarial case, it leads to minimum regret, i.e., the user does almost as well as if they knew the best network in advance.  If, however, the situation is changing over time, and especially if it is changing in some predictable manner, then learning the average payoff of each arm is not productive.

Periodic, repetitive patterns are a particularly common type of dynamic behavior.  Take, for example, the problem of network selection.  Network behavior is often repetitive, with user density and network quality following regular patterns: for example, office WiFi networks have no users at night, their performance drops when workers arrive in the morning, and the performance improves again during lunch hour.  Other networks are clogged with streaming video during lunch hour and in the evenings.  Periodic patterns are ubiquitous.

Unfortunately, bandit algorithms will fail badly in the case of periodic behavior.
As an example, suppose a player is playing a slot machine with two arms. The first arm gives a reward of $1$ when pulled on odd-numbered hours and $0$ otherwise, while the second arm does the reverse, with a reward of $1$ on even-numbered hours and $0$ otherwise.  In this simple case, a bandit algorithm will never learn this pattern, instead converging to the best single-action policy; and the best policy can only reap half of the maximum reward.  The player will receive an average payout of only $1/2$ per selection, despite a very predictable pattern.  And when this case is extended to cycle among $K$ arms, the best fixed choice of arm gives only $1/K$ of the total obtainable reward.  Thus, algorithms like EXP3 that minimize the regret do not guarantee good performance on periodic problems. 

\subsection{Contributions}
Our goal in this paper is to develop an efficient adversarial bandit algorithm for periodic settings, and to demonstrate the effectiveness of this algorithm in the context of the wireless network selection problem, yielding a new approach to network selection in dynamic, periodic environments.
The first step is to establish the right metric by which to evaluate bandit algorithms.  The performance of an adversarial bandit algorithm is heavily characterized by the definition of ``regret,'' which forms the baseline that it competes against.  And traditionally, the regret is computed with respect to the best fixed strategy.

For the periodic bandit setting, we define a better performance measure, \enquote{periodic regret}, which compares an algorithm's performance against the best periodic choice of arms.  No choice of period may match the input data perfectly, but the goal of periodic regret is to compare against the best choice.  
Moreover, we provide a generalized notion of periodicity, so that this notion of periodic regret can capture different types of patterned behavior. 

Next, we develop an algorithm that minimizes periodic regret, \algoname, a computationally efficient variant of EXP4 (Exponential-weight algorithm for Exploration and Exploitation using Expert advice)~\cite{auer2002nonstochastic}.  We show that the algorithm minimizes periodic regret in the following sense: 
with $K$ arms, $|F|$ possible periods, with each possible period of at most length $P$, then in an execution of length $T$ the periodic regret is at most $O\big(\sqrt{PKT \log K  + KT \log |F|}\big)$. We also prove a lower bound of $\Omega\big(\sqrt{PKT + KT \frac{\log |F|}{\log K}} \big)$ on periodic regret in an adversarial setting, showing that this is optimal within log-factors. 
An important aspect of \algoname is that it is a polynomial time algorithm: we leverage the structure provided by the target periodic patterns to reduce the computational complexity.  This is in contrast to EXP4 which requires exponential time and space in this context. 

The other major contribution of this paper is a new algorithm for network selection that is especially optimized for environments with periodic, patterned behaviors. We simulate the network selection problem, comparing \algoname to EXP3 and to a ``randomized optimal'' omniscient solution.  (We have previously seen in~\cite{appavoo2018shrewd} that these types of simulations are reasonably predictive of real-world behavior.)  

Our first observation is that \algoname does in fact efficiently learn periodic patterns and adapts relatively quickly to changes in network data rates (both discrete and continuous).  
We also see that \algoname does indeed outperform EXP3 in periodic settings, as expected, potentially yielding significant real-world improvements.  

Our second question involved the robustness of \algoname to noisy patterns. Real-world periodic patterns are rarely perfectly periodic, suffering noise and variance.  We experiment with noisy patterns, and see that \algoname continues to work well.

Finally, our third set of experiments looked at the performance of \algoname in the context of user mobility.  We simulate several scenarios where users change location over time, leading to changes in which networks they can access (and hence changes in the load on those networks).  For example, we imagine a typical office scenario where users arrive at the office in the morning, take a break for lunch, return to work, and then head home at the end of the day.  We observe that \algoname can also learn this type of periodic behavior, again, learning to adapt the users' network selection in a near-optimal fashion.  In fact, we compare two versions of the algorithm: one in which the algorithm is notified when networks become unavailable, and one in which it is not---we observe that even in the latter case where it is completely oblivious to the changes, the user strategy converges to near-optimal choices. 

Overall, we conclude that periodic adversarial bandit algorithms may have significant value, that \algoname is an efficient algorithm for the problem, and that it yields a potentially interesting and useful approach to network selection.

\section{Related work} \label{section:relatedWorks}
In this section, we discuss relevant work done on bandit algorithms, and state-of-art wireless network selection approaches.
Multi-armed bandit techniques have been successfully applied to wireless network selection \cite{appavoo2018shrewd, appavoo2019cooperation,multiuserlaxcomms}. They have also been considered for other resource selection problems, such as channel selection \cite{gai2010learning, tekin2011performance}, selection of the right sensors to query in a sensor network \cite{golovin2010online}, and selection of replica server for content distribution networks \cite{tran2014qoe}.

Many variations of bandit problems have been studied, in both stochastic and adversarial settings. EXP3 is the most well-known algorithm for the standard adversarial bandit problem. With $K$ arms and $T$ time steps, it establishes a pseudo-regret upper bound of $O(\sqrt{KT\log K})$, which almost matches the lower bound of $\Omega(\sqrt{KT})$ \cite{auer2002nonstochastic}.
The $\log K$ gap in the bounds has been recently closed by \cite{audibert2009minimax} bringing the upper bound down to $O(\sqrt{KT})$. But, these bound the regret against the best single-action policy, limiting their usefulness in a periodic setting.

A related problem is that of bandits with expert advice, defined in the same paper \cite{auer2002nonstochastic}. 
It defines a more general notion of regret, by competing against the best policy from a set.
With $K$ arms, $T$ time steps and $N$ experts, the EXP4 algorithm gives a pseudo-regret bound of $O(\sqrt{KT\log N})$. 
However, its  possibly high running time and memory cost limit its use in practice.
There are other algorithms for bandits with expert advice, like Context-FTPL. The latter is more computationally efficient, but has a weaker regret bound \cite{syrgkanis2016efficient}.
A lower bound of $\Omega(\sqrt{KT\frac{\log N}{\log K}})$ \cite{seldin2016lower} has been shown, but the $\log K$ gap in bounds has not been closed.

An equivalent formulation of our generalized periodic regret (explained later in Section \ref{section_generalizedperiodicregret}) has been briefly discussed in \cite[Chapter~4.2.1]{bestcontextset}, phrased as a contextual bandit problem where the algorithm competes against the best context set from a class of context sets.
The possible use of EXP4 is mentioned, but an alternative algorithm with a weaker regret bound is instead discussed as it has a reasonable polynomial-time performance unlike EXP4.

While much of the existing literature assume a single best arm, there are other efforts to look beyond this. One approach to the stochastic version of the problem is to allow reward distributions of the arms to occasionally change \cite{nonstationary2014, nonstationary2017}. Our work on the other hand is fully adversarial, and makes no assumptions on the rewards produced by the adversary.

%
Numerous wireless network selection approaches have been proposed. Some are centralized  \cite{aryafar2017max, bejerano2004fairness, mishra2006client, sui2016characterizing}; hence, not scalable and limited to managed networks. 
A number of distributed approaches have been proposed, with various limitations. Some rely on coordination from networks \cite{kauffmann2007measurement}, while others require cooperation of wireless devices \cite{deng2014all}. Others assume global knowledge \cite{niyato2009dynamics,aryafar2013rat, monsef2015convergence}, or availability of some information \cite{zhu2010network,cheung2017congestion}. A continuous-time multi-armed bandit approach in a stochastic setting has been considered in \cite{wu2016traffic}. A similar setting to ours, though non-periodic and in the stochastic setting, is considered in \cite{multiuserlaxcomms}.

\section{Wireless Network Selection} \label{section:problemFormulation}
Here, we describe the wireless network selection problem, discuss the periodicity of events in wireless environments, and formulate the network selection problem as a bandit problem.

\subsection{Wireless network selection problem.}
We consider an environment with multiple wireless devices and heterogeneous wireless networks, such as the one depicted in Figure~\ref{figure:heterogeneousNetworks}. The latter illustrates four mobile users with their (active) mobile devices, and five wireless networks, namely four WiFi networks and a cellular network (represented using 3 cellular base stations). The wireless networks have limited areas of coverage. Hence, each mobile device may have access to a different set of wireless networks depending on their location, e.g. different networks are available at home and at the office.
The bandwidths of wireless networks may also vary with time.
Each mobile device aims to quickly identify and associate with the best network, which may vary over time, to maximize their data rates. 

\begin{figure}[!htb]
\begin{center}
 \includegraphics [scale=0.6]
 {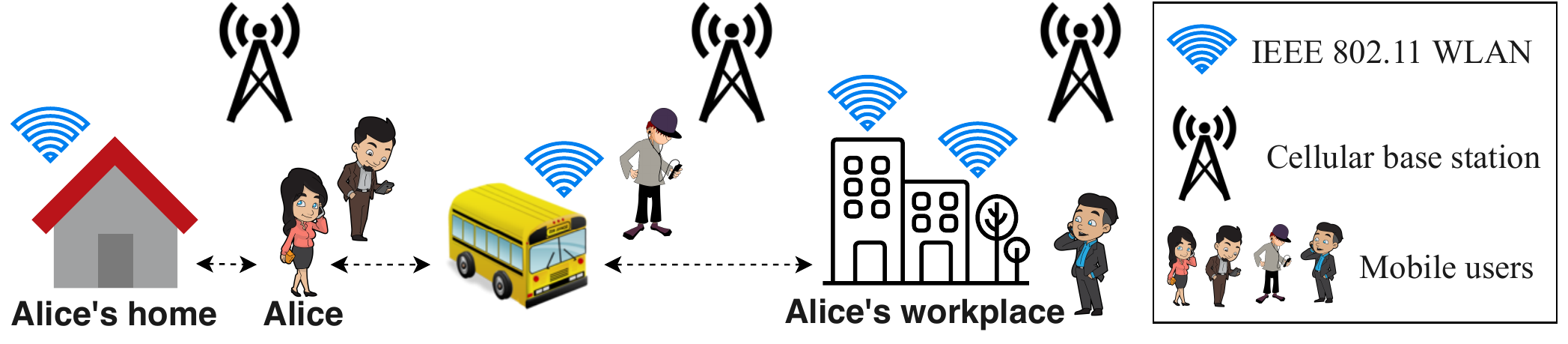}
\end{center}
\caption{Mobile devices with access to a different set of wireless networks as the user moves.}
\label{figure:heterogeneousNetworks}
\end{figure}

Mobile users tend to have daily routines that follow repetitive patterns---going to the office each morning, lunch at noon, returning home in the evening; these activities are performed at fixed times each weekday. Figure \ref{figure:heterogeneousNetworks} broadly depicts the daily routine of a mobile user, Alice. Network behavior, which is affected by user density, is also often repetitive and follows a regular pattern. For example, the available bandwidth of office WiFi networks is likely to be higher during lunch hours, where the office is nearly empty. A good network selection protocol learns and adapts to periodic patterns in network quality for better performance.

\subsection{Wireless network selection as a bandit problem}
A device must be aware of the bit rate it can observe from each network to perform an optimal network selection. While this information is unknown at the time of selection, the device can estimate the achievable bit rate by exploring the networks.  
The network selection problem can be seen as a multi-armed bandit problem in a multi-player setting. A mobile device is a player, and each network can be considered as an arm. Every so often (e.g. once per minute), a device selects a network (analogous to pulling an arm) and observes a bit rate (gain) for that network. The gain from other networks is unknown to the device.
Given that mobile devices operate in a dynamic environment, they must continuously explore and adapt to changes, by deciding which networks to select in sequence. The goal of each device is to maximize its cumulative gain over time. Since the quality of a wireless network is affected by its number of clients, other mobile devices in the environment may be considered to be adversaries. We hence use the adversarial setting. A leading bandit algorithm in this setting is EXP3.

\section{Periodic Bandit Problem} \label{section:periodicProblem}
In this section, we introduce the periodic bandit problem and discuss periodic regret.

We consider a general bandit problem.
On each time step, an algorithm is allowed to pick any one out of $K$ possible arms, and each arm produces a certain amount of reward. These rewards are unknown to the algorithm, which can only observe the reward of the arm it picked. We aim to maximize the total reward obtained by the algorithm.
We study the adversarial setting with a possibly adaptive adversary, which decides on the distribution of rewards at each time step, taking into consideration the outcomes of past random events.

Let $K$ be the number of arms. The set of arms is $[K] := \{1,2,\cdots,K\}$.
Let $x_i(t) \in [0,1]$ be the reward earned by arm $i \in [K]$ at time step $t$.
Let $a(t) \in [K]$ be the arm played by the algorithm at time $t$.
Let $T$ be the total number of time steps. The set of time steps is $[T] := \{1,2,\cdots,T\}$.
Thus, the total reward earned by the algorithm after $T$ iterations is $\sum_{t=1}^T x_{a(t)}(t)$.
The commonly used performance measure for bandit algorithms is regret.
Regret compares the total reward obtained by the algorithm against a ``best possible'' reward ``OPT'' after some number of time steps $T$. Different types of regret compare the algorithm's result to different notions of the optimal result.

We can define a form of regret where OPT is allowed to pick any arm in $[K]$ at each time step. For later reference we will refer to this as full regret, defined as follows:
\begin{align*}\displaystyle  R_{full}(T) = \sum_{t=1}^T \max_{i \in [K]} E\Big[x_i(t)\Big] - E\Big[\sum_{t=1}^T x_{a(t)}(t)\Big] \end{align*}

The above definition uses what is commonly known as pseudo-regret, rather than expected regret. For the rest of this paper, we will often refer to pseudo-regret as simply ``regret''. Expectations are taken over the possible randomness of the algorithm and adversary.

In most studies of adversarial bandits, a weaker definition of regret is used. This is because full regret uses too powerful an adversary, and it is impossible to achieve better than linear expected full regret in the worst case (we include a proof in 
Appendix \ref{appendix:lowerboundonfullregret}).
Therefore, it is common to define a notion of regret where OPT is required to use the same arm for all $T$ time steps. We refer to this as weak regret, defined as follows:
\begin{align*}\displaystyle  R_{weak}(T) = \max_{i \in [K]}\sum_{t=1}^T E\Big[x_i(t)\Big] - E\Big[\sum_{t=1}^T x_{a(t)}(t)\Big] \end{align*}

Weak regret however, severely limits what OPT can do, and being competitive with an algorithm that can only pick one arm and stick to it may not be a very strong result.

\subsection{Periodic Regret}

We can bridge the two with a periodic definition of regret. Taking the idea that a periodic choice of arms is likely to perform well in situations with periodic patterns, we can define a regret function which measures how competitive an algorithm is with the best periodic choice of arms.
For example, we can say OPT is forced to play the same arm every $\tau \in \mathbb{N}$ steps. This defines a regret function as follows, 
\begin{align*}\displaystyle  R_\tau(T) = \sum_{\ell=1}^\tau \max_{i \in [K]}\sum_{\tbar=0}^{\floor{\frac{T-\ell}{\tau}}} E\Big[x_i(\tbar\tau + \ell)\Big] - E\Big[\sum_{t=1}^T x_{a(t)}(t)\Big] \end{align*}
As OPT may optionally still pick the same arm on all time steps, this is a generalization of weak regret. This makes for a regret value in between weak regret and full regret. 

If we were competing against the regret for a specific, known value of $\tau$, this would be equivalent to playing $\tau$ independent instances of the adversarial bandits problem over approximately $T/\tau$ time steps each. By playing $\tau$ separate instances of an algorithm for weak regret, and by Theorem \ref{theorem_partitionlowerbound} in Section \ref{subsection:lower_bound_single_partition}, we have an upper/lower bound of $\Theta(\sqrt{\tau K T})$.

However, if we were to consider that the ``best possible'' period $\tau$ may not be known (for example, if OPT were to consist of the best periodic function for any of the possible periods $\tau \in \{1,\cdots,P\}$), these bounds do not apply as easily.

\subsection{Generalized Periodic Regret}
\label{section_generalizedperiodicregret}

A generalization of the periodic case is the use of partition functions. Fix a maximum number of labels $P$. We define this upper bound $P$ for use in our analysis later on. A partition function $f : [T] \to [P]$ is a function that assigns every time step a label from $1$ to $P$.
We consider two partition functions the same if their choice of label assignments are permutations of each other.
The regret under function $f$ would be when OPT is forced to play the same arm for all timesteps with the same label as assigned by $f$.
\begin{align}\displaystyle  R_f(T) = \sum_{\ell \in [P]} \max_{i \in [K]} \sum_{\tbar \in f^{-1}(\ell)} E\Big[x_i(\tbar)\Big] - E\Big[\sum_{t=1}^T x_{a(t)}(t)\Big]
\label{expr:single_function_periodic_regret}
\end{align}

Consider a set of partition functions $F \subseteq \{f : [T] \to [P]\}$ for some $P \in \mathbb{N}$. $F$ is necessarily finite. The regret under the function set $f$ would be when OPT can choose to play using any of the partition functions in $F$. This gives the following regret definition:
\begin{align}\displaystyle  R_F(T) = \max_{f \in F}\sum_{\ell \in [P]} \max_{i \in [K]} \sum_{\tbar \in f^{-1}(\ell)} E\Big[x_i(\tbar)\Big] - E\Big[\sum_{t=1}^T x_{a(t)}(t)\Big]
\label{expr:generalized_periodic_regret}
\end{align}

This definition (\ref{expr:generalized_periodic_regret}) of periodic regret gives us more choice in how we want to define our potential periodic patterns to learn, through deciding on the labels on each time step for each function. We demonstrate this with our choice of partition functions in Section \ref{section:evaluation}.

To model the example described earlier with periods $\tau \in \{1,2,\cdots,P\}$, we can use the set of partitions $F = \{f_1,f_2,\cdots,f_P\}$, where $f_\tau(t) := (t \text{ mod } \tau) + 1 $ for each $t \in [T]$, $\tau \in [P]$.

\section{The \algoname Algorithm} \label{section:proposedAlgorithm}          We discuss the relationship between our generalized periodic setting and the problem of bandits with expert advice, and hence the applicability of EXP4 \cite{auer2002nonstochastic} to the problem. We use this to introduce \algoname, an efficient algorithm for generalized periodic regret.

\subsection{Applying Bandits with Expert Advice to Periodic Bandit Problems}


Periodic bandit problems can be reduced to the problem of bandits with expert advice. In the problem of bandits with expert advice, we are given a set $\Pi$ of $N$ experts. Each expert predicts an arm on each time step. We fix the number of time steps $T$. Thus an expert can be seen as a function $\pi: [T] \to [K]$. An algorithm to solve this problem would make use of each expert's predictions on each time step, to obtain a reward competitive with the best expert in the set. This gives us the following regret definition:
\begin{align*}\displaystyle  R_\Pi(T) = \max_{\pi \in \Pi}\sum_{t=1}^T x_{\pi(t)}(t) - E\Big[\sum_{t=1}^T x_{a(t)}(t)\Big] \end{align*}
This can be used to model all of the above notions of regret. For full regret, we have $\Pi := \{\pi: [T] \to [K]\}$, the set of all possible functions from $[T]$ to $[K]$. For weak regret, $\Pi$ is the set of all constant functions from $[T]$ to $[K]$.

In the generalized periodic setting, let $F$ be the set of partition functions $f : [T] \to [P]$. For each function $f \in F$, let $\Theta_f$ be the set of all possible mappings $\theta: f([T]) \to [K]$ from the image set $f([T])$ of $f$ to the set of arms $[K]$ (thus $|\Theta_f| = K^{|f([T])|}$). Each composition $\theta \circ f$, $f \in F$, $\theta \in \Theta_f$ thus represents a possible mapping of the time steps $[T]$ to arms. Thus, for the generalized periodic setting, $\Pi = \{\theta \circ f \mid f \in F, \theta \in \Theta_f\}$.

We note that when $\Pi_1 \subseteq \Pi_2$, we will have $R_{\Pi_1}(T) \leq R_{\Pi_2}(T)$. Let $\Pi_{full}$, $\Pi_{weak}$ and $\Pi_{F}$ be the sets of functions corresponding to full regret, weak regret and generalized periodic regret under some function set $F$ respectively. Thus, for any nonempty set $F$ of partition functions, we have $R_{\Pi_{weak}}(T) \leq R_{\Pi_{F}}(T) \leq R_{\Pi_{full}}(T)$.

An existing algorithm for this problem is the EXP4 algorithm \cite{auer2002nonstochastic}, which achieves a regret upper bound of $O(\sqrt{KT\log N})$, where $N := |\Pi|$. We can thus apply EXP4 directly to our problem. However, a commonly cited drawback of the EXP4 algorithm is that its running time and memory cost are at least linear in $N$. This is an issue as $N$ is often very large. For example, in the generalized periodic setting, the size of $N$ could easily be on the order of $|F|K^P$, which is exponential in $P$. However, we show below that in the generalized periodic setting, we can devise an algorithm that is distributionally equivalent to EXP4 and can be made to run in time polynomial in $|F|$, $K$ and $P$.

The EXP4 algorithm works by assigning a weight $w_\pi$ (with initial value $1$) to each expert $\pi \in \Pi$. The probability $p_i(t)$ of playing an arm $i \in [K]$ would then be ${\sum_{\pi(t)=i}w_\pi(t)}/{\sum_\pi w_\pi(t)}$, the ratio of the combined weights of the experts agreeing to play arm $i$ to the total weight of the experts. Whenever an arm $i \in [K]$ is played, each expert who suggested arm $i$ will have their weight adjusted by some factor $\exp(\frac{\gamma}{K}x_i(t)/p_i(t))$.
More details on EXP4 are given in \cite{auer2002nonstochastic}.
Note that it discusses a more general form of expert advice where each expert suggests a probability vector on the arms. However, we only require the case where at each time step, each expert suggests one arm with probability $1$, and all other arms with probability $0$.

\subsection{\algoname, Memory and Running Time Costs}
\algoname (Algorithm \ref{alg:groupedexp4}) is distributionally equivalent to the EXP4 algorithm when run with the set of experts $\Pi = \{\theta \circ f \mid f \in F, \theta \in \Theta_f\}$. The key intuition behind this algorithm is that the generalized periodic setting produces many symmetries in the weight computation for each expert. Specifically, we take advantage of how for each partition function $f$, the set of experts contains every possible combination of arm assignments to labels in the image set $f([T])$. This allows us to compute the probabilities that EXP4 would play each arm at each time step without computing the individual weights of every expert.
\begin{algorithm}[t!]
\caption{\algoname} \label{alg:groupedexp4}
\begin{algorithmic}[1]
\Procedure{Initialization}{}
\ForEach {$f \in F$}
    \ForEach {$\ell \in f([T])$}
        \ForEach {$i \in [K]$}
            \State Initialize $b_i^{\ell,f}(1) = 1$
        \EndFor
    \EndFor
\EndFor
\EndProcedure

\Procedure{Algorithm}{}
\ForEach {time step $t = 1,2,\cdots,T$}
    \ForEach {$i \in [K]$}
        \State $\ds r_i(t) := \sum_{f \in F} \Big( b_i^{f(t),f} (t) \prod_{\ell \in f([t]) \setminus \{f(t)\}}\sumjk b_j^{\ell,f}(t) \Big)$
    \EndFor

    \ForEach {$i \in [K]$}
        \State $\ds p_i(t) = \frac{r_i(t)}{\sumjk r_j(t)}$
    \EndFor

    \State Play arm $i_t \in [K]$ from the probabilities $p_1(t), p_2(t), \cdots, p_K(t)$
    \State Obtain reward $x_{i_t}(t)$

    \ForEach {$f \in F$}
        \ForEach {$\ell \in f([T])$}
            \ForEach {$i \in [K]$}
                \If {$i = i_t$ and $\ell = f(t)$}
                    \State $b_i^{\ell,f}(t+1) = b_i^{\ell,f}(t) \exp(\gammK x_i(t)/p_i(t))$
                \Else
                    \State $b_i^{\ell,f}(t+1) = b_i^{\ell,f}(t)$
                \EndIf
            \EndFor
        \EndFor
    \EndFor
\EndFor
\EndProcedure
\end{algorithmic}
\end{algorithm}

For brevity, let $P_f := |f([T])|$ be the number of labels used by the function $f$. Necessarily $P_f \leq P$.
The memory requirement is $O(K\sum_{f \in F} P_f)$, which is at most $O(KP|F|)$. A naive implementation of the algorithm gives a running time of $O(K^2 \sum_{f \in F} P_f)$ per time step, but with some pre-computation, the running time can be lowered as shown in
Appendix \ref{appendix_optimizedperiodicexp4}.

\subsection{Correctness of \algoname}
\label{proofOfCorrectness}

To show correctness, we show that our algorithm produces the same probability distribution over arms as EXP4 in every time step.
Define $\pi_{\theta,f}$ as the expert which at time $t$ recommends arm $\theta\circ f(t)$ with probability $1$ and all other arms with probability $0$. We show this algorithm is distributionally equivalent to EXP4, where $\Pi = \{\pi_{\theta, f} | f \in F, \theta \in \Theta_f\}$.
In EXP4, each expert $\pi_{\theta,f}$ would have some weight $w_{\theta,f}(t)$ at time step $t$. At time step $t$, EXP4 plays arm $i$ with probability $p_i(t)$ represented by the following expression:
\begin{align*}
\ds p_i(t) = \frac{\sum_{f \in F, \theta \in \Theta_f, \theta\circ f(t) = i} w_{\theta,f}(t)}{\sum_{f \in F, \theta \in \Theta_f} w_{\theta,f}(t)}
\end{align*}
Thus, to show that the two algorithms are distributionally equivalent, as $p_i(t) := {r_i(t)}/{\sum_{j=1}^{K}r_j(t)}$ in our algorithm, for each successive time step $t$, we only need to show the following:
\begin{align*}
\ds r_i(t) = \sum_{f \in F, \theta \in \Theta_f, \theta\circ f(t) = i} w_{\theta,f}(t)
\end{align*}
The details of this derivation is given in
Appendix \ref{appendix_proofofcorrectness}.
We can thus formally state a regret upper bound as follows (Theorem \ref{theorem:upper_bound}). This upper bound comes directly from EXP4's regret bound of $O(\sqrt{KT\log N})$, where the number of experts $\ds N = \sum_{f \in F}K^{|f([T])|} \leq |F|K^P$.

\begin{theorem}
\label{theorem:upper_bound}
With $K$ arms, $T$ time steps, $|F|$ partition functions, with every function having at most $P$ labels, \algoname gives a regret upper bound of $O\big(\sqrt{PKT \log K  + KT \log |F|}\big)$.
\end{theorem}

\section{Lower Bounds} \label{section:theoreticalAnalysis}
In this section, we provide lower bounds for the case of a single partition and for a set of partitions. We demonstrate that the upper and lower bounds differ by a factor of $\log K$.

The existing regret lower bound for the problem of bandits with expert advice \cite{seldin2016lower} is $\Omega\big(\sqrt{KT\frac{\log N}{\log K}}\big)$. This lower bound is derived by dividing the time steps $[T]$ into $\frac{\log N}{\log K}$ equal parts.
For the generalized periodic setting, as this lower bound uses an instance that can be modeled with a single partition function, it does not give immediate insight into whether having multiple different periods or partition functions increases the difficulty of the problem.

\subsection{Lower Bound for a Single Partition}
\label{subsection:lower_bound_single_partition}
We consider the case with only a single partition function $f : [T] \to [P]$, which partitions the time steps into $P$ labels $1,2,\cdots,P$. The sizes of the partitions are $|f^{-1}(1)|,$ $|f^{-1}(2)|,$ $\cdots,$ $|f^{-1}(P)|$ respectively. It seems like intuitively, by seeing this as $P$ separate instances of the weak regret setting, and by the existing $\Theta(\sqrt{KT})$ upper/lower bounds on weak regret \cite{auer2002nonstochastic, audibert2009minimax}, we would have an upper/lower bound of $\Theta(\sum_{\ell = 1}^P \sqrt{K |f^{-1}(\ell)|})$. For equally sized partitions of size approximately $\frac{T}{P}$ each, this bound would be $\Theta(\sqrt{PKT})$.

However, while the upper bound is clearly met by running $P$ independent instances of an algorithm for weak regret, the lower bound is less clear. Even when considering it as $P$ separate instances, there is a possibility of an algorithm ``reacting'' to losses in other instances to play differently in the current instance, obtaining a higher total reward as a result. For completeness, we include a proof for the lower bound (Theorem \ref{theorem_partitionlowerbound}) in
Appendix \ref{appendix_partitionlowerbound}

\begin{theorem}
\label{theorem_partitionlowerbound}
Fix a partition function $f:[T] \to [P]$ which assigns a label to each time step.
Assume that for each $\ell \in f([T])$, there are at least $K/(4\ln \frac{4}{3})$ time steps with label $\ell$.
Then the minimax pseudo-regret (\ref{expr:single_function_periodic_regret}), over all algorithms $a$ and adversaries $R$, has a lower bound as follows, for some positive constant $c$: 
\begin{align*}
\ds \inf_a \sup_R \Big( \max_{\theta \in \Theta_f} E\Big[\sum_{t \in [T]} x_{\theta \circ f(t)}(t)\Big] - E\Big[\sum_{t \in [T]} x_{a(t)}(t)\Big] \Big) \geq \sum_{\ell \in [P]}\sqrt{cK|f^{-1}(\ell)|}
\end{align*}
\end{theorem}
If we consider the simple case where OPT may play only periodic functions from any period $\tau \in \{1,2,\cdots,P\}$, it can do no worse than if it were only allowed to play at period $P$. We thus obtain a lower regret bound of $\sqrt{PKT}$.

\subsection{Lower Bound for the Generalized Periodic Setting}
\label{subsection:general_lower_bound}
Let $F$ be the set of partitions, so $|F|$ is the number of partitions. Let $P$ be the maximum number of labels of any partition in $F$.
For sufficiently large $T$ and $K \leq P$, we obtain a pseudo-regret(\ref{expr:generalized_periodic_regret}) lower bound of $\Omega(\sqrt{PKT} + \sqrt{KT \frac{\log |F|}{\log K}})$.
It is proved in
Appendix \ref{appendix_generalizedlowerbound}

If $P < K$ instead, a simple lower bound can be obtained by using only $P$ out of the $K$ arms, so we obtain a problem with $P$ arms and maximum partition size $P$. This gives us a lower bound of $\Omega\Big(\sqrt{PKT} + \sqrt{PT \frac{\log |F|}{\log P}}\Big)$.
We can then merge these two lower bounds into a single expression
$\Omega\Big(\sqrt{PKT} + \sqrt{\min(P,K)T \frac{\log |F|}{\log \min(P,K)}}\Big)$.


\subsection{Analysis of Bounds}
A conclusion we can make from Section \ref{subsection:general_lower_bound} is that having multiple periods indeed increases the difficulty of the problem - we have obtained a lower bound higher than the known upper bound of $O(\sqrt{PKT})$ had only one partition function of the maximum period $P$ been used.

With $K$ arms, $T$ time steps, $|F|$ partition functions, with every function having at most $P$ labels, \algoname gives an upper bound of $O\big(\sqrt{PKT \log K  + KT \log |F|}\big)$.
On the other hand, we have a lower bound of $\Omega\big(\sqrt{PKT + KT \frac{\log |F|}{\log K}} \big)$ in the case where $K \leq P$. This gives a gap of $\sqrt{\log K}$ between the two bounds.
Interestingly, this log-factor is the same as the current gap between the upper and lower bounds in the problem of bandits with expert advice. This is possibly because we use a similar lower bound proof to the problem of bandits with expert advice \cite{seldin2016lower}, as well as a similar algorithm for the upper bound.

\section{Experimental Evaluation} \label{section:evaluation}
In this section, we discuss the implementation details of \algoname and parameter values chosen, evaluate the algorithm via simulation, and compare its performance to EXP3 \cite{auer2002nonstochastic}. We show how \algoname (a) learns periodic patterns over time under both discrete and continuous changes in network data rates, (b) outperforms EXP3,
(c) is robust to noisy patterns, and (d) adapts to changes due to mobility of users.

We benchmark against ``Optimal Random'', a player with prior knowledge of the actual bandwidths of each network.
In each time slot, it picks a network from a probability distribution equal to the ratios of the bandwidths. For example, with network bandwidths $4,10$ and $6$, the probability of picking the networks will be $0.2$, $0.5$ and $0.3$, respectively.

All the algorithms are implemented in Python, using SimPy \cite{simpy}, while the core algorithm is written in C++.
We use a time-varying learning rate $\gamma = t^{-\frac{1}{10}}$ \cite{maghsudi2013relay} for both \algoname and EXP3; $\gamma$ slowly tends to zero to ensure convergence \cite{tekin2011performance} while at the same time ensures that the algorithm does not take too long to learn (it learns slowly when $\gamma$ is very small).
Although they are not pre-requirements of \algoname, for simplicity, we assume that (a) a network’s bandwidth is equally shared among its clients, and (b) devices are time-synchronized.
To reduce numerical error in our simulations, we substitute computations of $\sum_{x\in Y}\exp(x)$ with $\exp(\max_{x\in Y}x)$. 
In nearly all cases, sums of exponentials in our algorithm are heavily dominated by a single term, making the values of the two expressions approximately equal. Experimentally, we find that this has negligible effects on the values computed within the algorithm.

We do simulations on synthetic data.
We consider setups with 20 mobile devices and 3 wireless networks, unless otherwise specified. While the number of devices remain constant throughout the simulation run, the data rates and availability of networks may change.
We assume that a network selection is performed once every minute; hence, 1440 time slots is one simulated day. All results presented are from 20 simulation runs, of 86,400 time slots each (i.e., 2 simulated months). The pattern of network behavior and/or user mobility over the first 1440 time slots is repeated 60 times; we refer to each repetition as an \enquote{iteration}. 

We apply \algoname in the generalized periodic setting.
We define a partition function of period $\tau$ as one which divides each iteration of 1440 time slots into $\tau$ equal contiguous segments, labeled $1$ to $\tau$ in chronological order. The same labels are used for each successive repetition.
Unless otherwise specified, we use the period set $\{1, \cdots , 24\}$. This refers to using 24 partition functions, of periods $1$ to $\tau$ respectively.

\subsection{Evaluation Criteria}
Good assignments of devices to networks divide the available bandwidth evenly among the devices.
We thus evaluate the performance of the algorithms based on the lowest data rate observed by any of the devices. We compare this to the optimal allocation of devices, which maximizes the lowest data rate observed by any device. 
If a device with the lowest data rate observes 3Mbps, but the optimal's lowest is 5Mbps, we say it loses 40\% of its achievable gain.
We refer to this percentage loss as the ``distance to optimal minimum'' in our results.

We do not use average cumulative gain as a performance measure because in our problem setting, average gain is maximized as long as there is at least one user in each network.

\subsection{Performance Comparison of Algorithms}
We consider two setups, both at an office with two WiFi networks and a cellular network. The data rates of these networks vary over time.
The first setup involves discrete changes in network bandwidths at fixed time intervals (Figure \ref{subfigure:discrete_data_rate}).
In the second setup, the data rates vary continuously with time (Figure \ref{subfigure:continuous_data_rate}).
Figures \ref{subfigure:distanceToNE_discrete_data_rate} and \ref{subfigure:distanceToNE_continuous_data_rate} show that in both setups, the distance to optimal minimum of \algoname drops over time while EXP3 shows no noticeable improvement with time.

\begin{figure*}[!htb]
    \centering
    \begin{subfigure}[t]{0.5\textwidth}
        \centering
        \includegraphics [scale=0.7]
{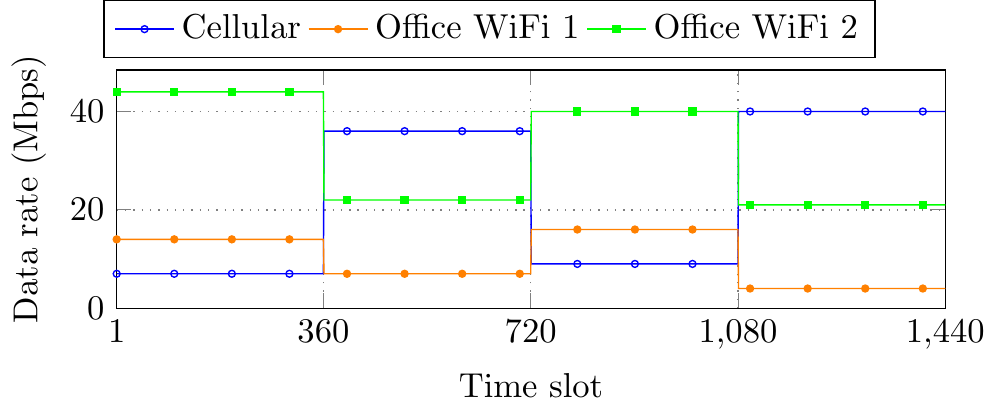}
        \caption{Discrete changes in network data rates.}
        \label{subfigure:discrete_data_rate}
    \end{subfigure}%
    ~ 
    \begin{subfigure}[t]{0.5\textwidth}
        \centering
        \includegraphics [scale=0.7]
{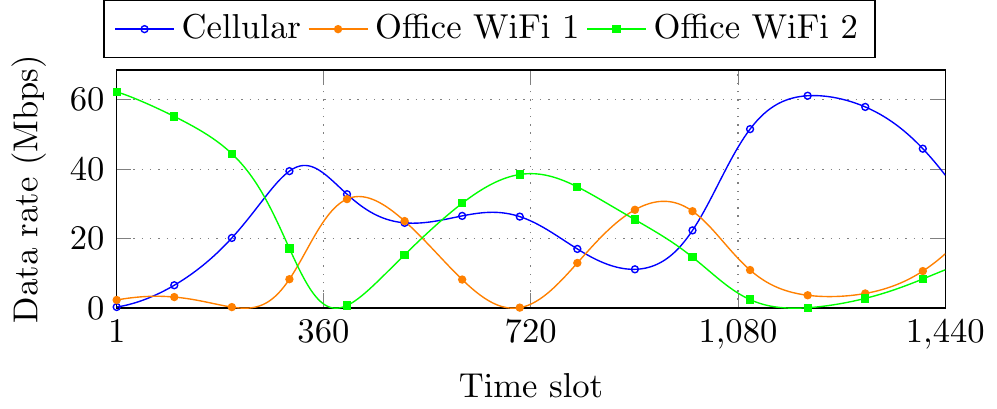}
        \caption{Continuous changes in network data rates.}
        \label{subfigure:continuous_data_rate}
    \end{subfigure}
    \caption{Changes in network data rates over one iteration (this is repeated 60 times).}
    \label{figure:learning_pattern_data_rates}
\end{figure*}

\begin{figure*}[!htb]
    \centering
    \begin{subfigure}[t]{0.5\textwidth}
        \centering
        \includegraphics [scale=0.7]
        {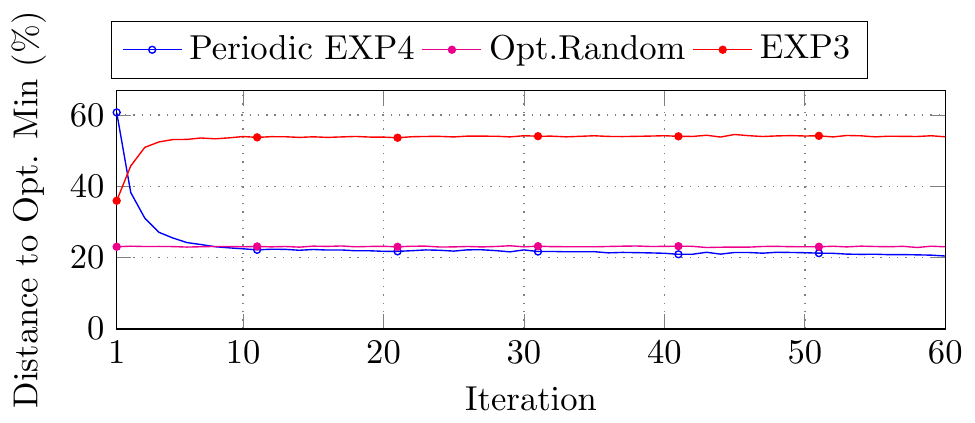}
        \caption{Performance under discrete setup.}
        \label{subfigure:distanceToNE_discrete_data_rate}
    \end{subfigure}%
    ~ 
    \begin{subfigure}[t]{0.5\textwidth}
        \centering
        \includegraphics [scale=0.7]
        {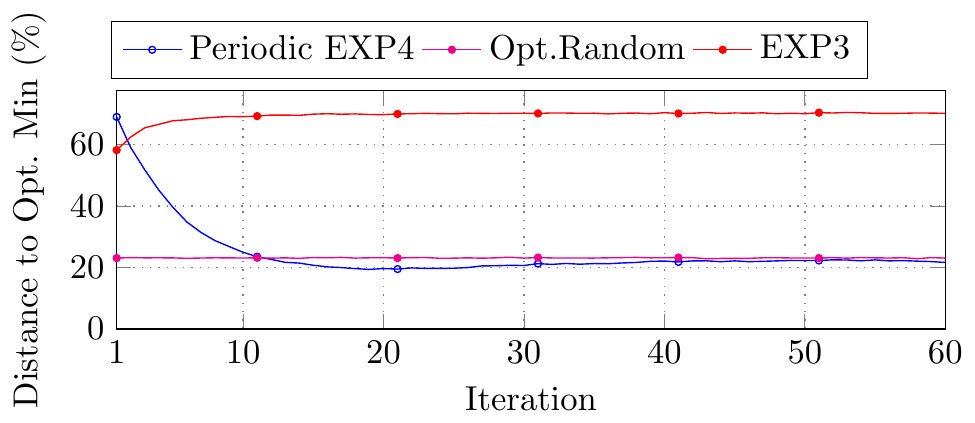}
        \caption{Performance under continuous setup.}
        \label{subfigure:distanceToNE_continuous_data_rate}
    \end{subfigure}
    \caption{Distance to optimal minimum of \algoname and EXP3 over 60 iterations.}
    \label{figure:learning_pattern}
\end{figure*}

\begin{figure*}[!htb]
    \centering
    \begin{subfigure}[t]{0.8\textwidth}
        \centering
        \includegraphics [scale=0.65]
{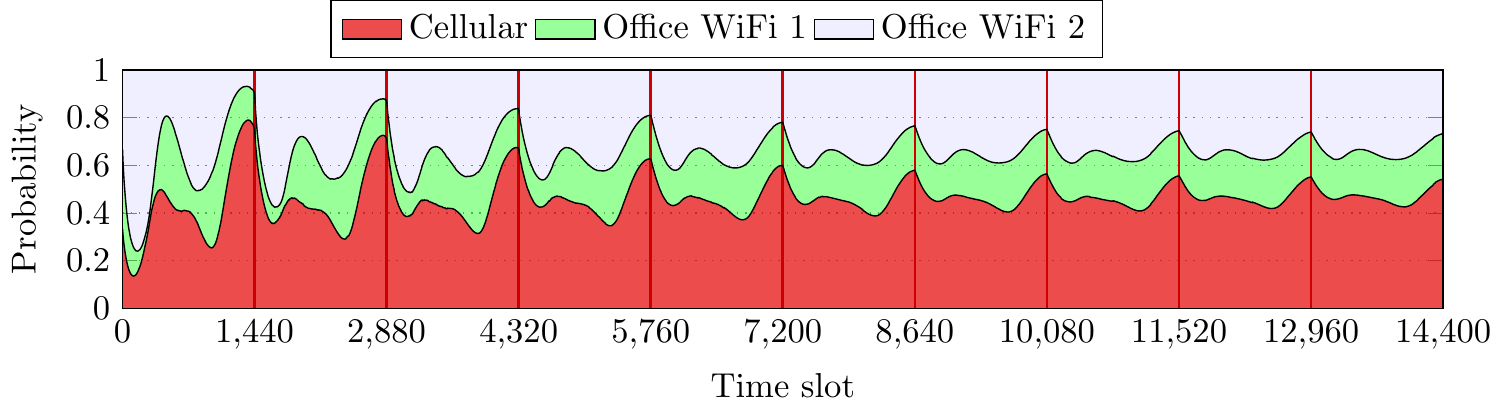}
        \caption{EXP3: Combined probabilities for each network over first 10 iterations.}
        \label{subfigure:total_probability_EXP3}
    \end{subfigure}%
    ~~
    \begin{subfigure}[t]{0.2\textwidth}
        \centering
        \includegraphics [scale=0.65]
{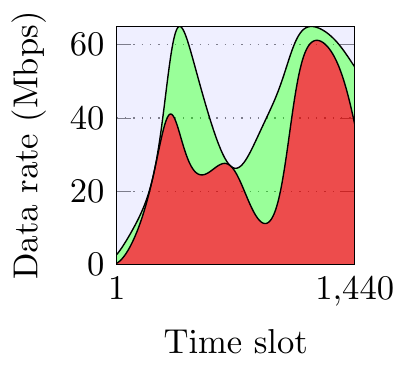}
        \caption{Bandwidth ratio}
        \label{subfigure:bandwidth_ratios}
    \end{subfigure}
    ~
    \begin{subfigure}[t]{\textwidth}
        \centering
        \includegraphics [scale=0.65]
{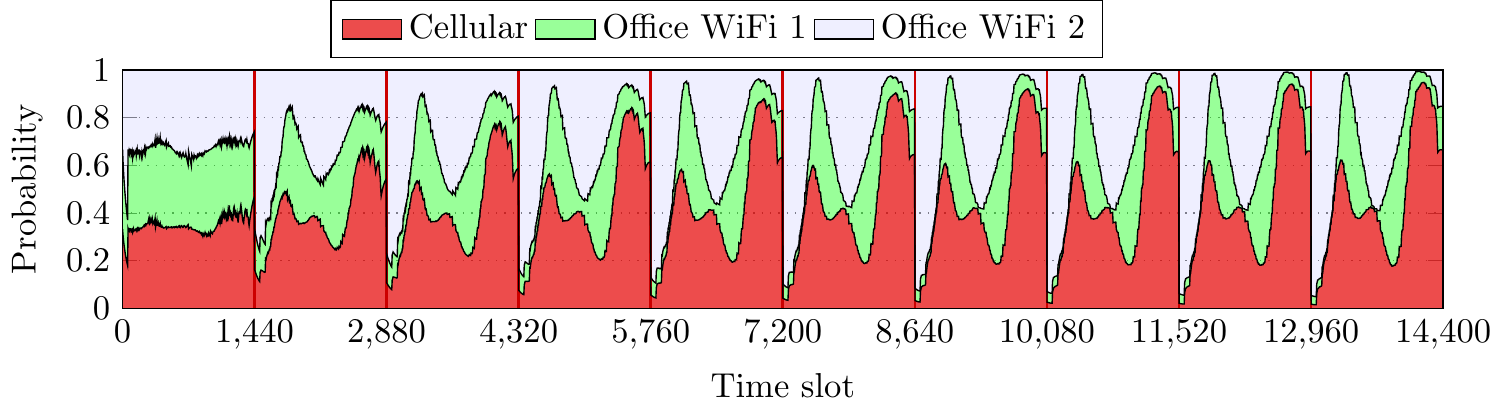}
        \caption{\algoname: Combined probabilities for each network over first 10 iterations.}
        \label{subfigure:total_probability_periodicEXP3}
    \end{subfigure}
    \caption{Area chart showing the time variation of combined probabilities in the continuous setup. Figure \ref{subfigure:bandwidth_ratios} shows the actual ratio of the bandwidths of the three networks within any one iteration.}
    \label{figure:total_probability_plots}
\end{figure*}

Figure \ref{figure:total_probability_plots} for the continuous setup explains this improvement. The figure for the discrete setup is in
Appendix \ref{subsection:learning_patterns}.
At each time step, each user has a probability of picking each of the networks.
If we consider the combined probability of picking each network, we can see that in \algoname, these probabilities converge towards the ratios of the bandwidths of the networks (Figures \ref{subfigure:total_probability_periodicEXP3}). This is despite the continuous setup having no obvious best period. On the other hand, EXP3's probabilities slowly flatten out (Figure \ref{subfigure:total_probability_EXP3}). This is consistent with what we would expect, as EXP3 seeks to be competitive with the best fixed-action policy, meaning that it only seeks out the best fixed arm to play.

\begin{figure}[!htb]
\begin{center}
\includegraphics [scale=0.9]
{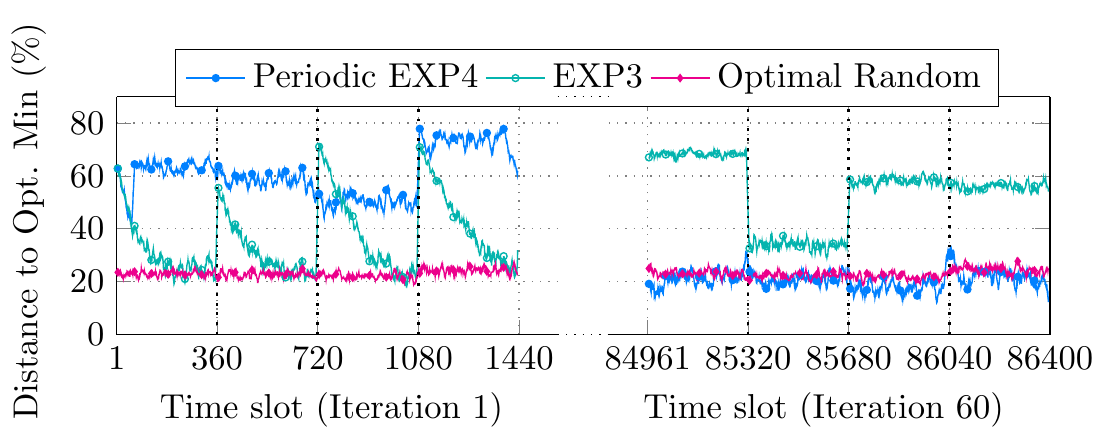}
\end{center}
\caption{Distances to Optimal minimum in the first and last repetitions of the discrete setting in Figure \ref{subfigure:discrete_data_rate}. Vertical lines indicate points where data rates change.}
\label{figure:algo_comparison_distanceToNashEquilibrium}
\end{figure}

Figure \ref{figure:algo_comparison_distanceToNashEquilibrium} shows that while EXP3 initially learns more quickly, \algoname eventually outperforms EXP3 (which converges to the network with the best average performance), with a performance similar to Optimal Random.
From our experiments, we find that while all algorithms have similar total cumulative gains, we may note that \algoname is fairer than EXP3, with significantly lower variance. We present these results in
Appendix \ref{subsection:performance_comparison}.

\subsection{Other Experiments}
In
Appendix \ref{section:appendix_experiments},
we discuss a few more experiments, the results of which are briefly summarized as follows:
\begin{enumerate}
\item \textbf{Performance in Noisy Settings:}
On each time step, we apply a 10\% Gaussian noise to each of the networks' data rates. We find that our algorithms are largely unaffected by noise in the data, giving similar results with and without noise.

\item \textbf{Comparison of Period Sets:}
We do a comparison between different possible period sets $F$. We find that the algorithm learns more slowly with larger period sets (e.g. $\{1,2,\cdots,45\}$, as compared to $\{1,2,\cdots,15\}$), but can converge to better results on more complex instances (instances where the bandwidth may fluctuate more wildly).
    
\item \textbf{Mobility of Users:}
We consider a setup where users move around and have access to different sets of networks at different times.
We compare Vanilla \algoname, which is oblivious to networks possibly becoming unavailable, against an optimized version,
which selects only from the set of currently available networks. While the optimized version initially yields a better performance, they eventually perform equally well when the Vanilla \algoname algorithm learns the pattern.
\end{enumerate}


    



\section{Conclusion} \label{section:conclusion}
In this paper, we develop an efficient variant of EXP4 for the periodic bandit problem, give nearly matching upper and lower bounds for it, and demonstrate its advantages in learning periodic behavior in the context of the network selection problem.

An interesting issue raised in contrasting this paper and \cite{nonstationary2014, nonstationary2017} is whether non-stationary bandit problems are better modeled stochastically or adversarially.
While these papers address non-stationary rewards primarily in a stochastic setting with some adversarial aspects, we tackle the periodic bandit problem in a fully adversarial setting. 
Using the adversarial setting has the benefit of not placing any constraints on the adversary;
we adapt to the periodic setting only through our definition of regret.
A proper comparison of stochastic and adversarial methods for network selection is a possible future line of work.

%


\bibliography{ms}

\newpage
\appendix
\input{appendix}

\end{document}

%% file: appendix.tex
\section{Appendix: Theoretical Results and Proofs} \label{section:appendix}
    \input{appendix_theory}

\section{Appendix: Experiments} \label{section:appendix_experiments}
    \input{appendix_experiments}

%% file: appendix_theory.tex
\subsection{Lower bound on Worst Case Full Regret}
\label{appendix:lowerboundonfullregret}
We construct a proof using a deterministic oblivious adversary for full generality. Proofs using a randomized oblivious adversary or adaptive adversary are simpler. 
A deterministic adversary must select the full sequence of rewards prior to the first round. This is in contrast to an adaptive adversary, a more powerful adversary which is allowed to select rewards each round with full knowledge of the outcomes of random events occurring prior to the round.

Let $T$ be the number of time steps and $K$ be the number of arms.
Fix an algorithm $a$. We show that there exists a problem instance (a predetermined sequence of rewards for each arm) such that algorithm $a$ obtains an expected full pseudo-regret of at least $T\frac{K-1}{K}$.

We construct a problem instance which, for each time step $t$ from $1$ to $T$, has one arm $b_t \in [K]$ which gives a reward of $1$, while all other arms give a reward of $0$. We construct each $b_t$ based on the algorithm (but not on the algorithm's choices) inductively as follows:

At the start of the algorithm, the algorithm plays arms with probabilities $p_1,\cdots p_K$ respectively. Define $b_1$ to be the arm with the lowest probability of being played.

We maintain the following invariant with parameter $\tau \in [T]$ - when running the algorithm $a$ on the constructed problem instance from time steps $1$ to $\tau$, the expected total reward $E[\sum_{t=1}^{\tau}x_{a(t)}(t)]$ by the algorithm is at most $\frac{\tau}{K}$. With the definition of $b_1$ above, we can see that this invariant holds for $\tau=1$.

Now fix any later time step $\tau \in [T]$. The algorithm $a$'s decision on time step $\tau$ can only be based on past rewards and the sequence of arms played by the algorithm on time steps up to $\tau-1$. As our choices of $b_1$ to $b_{\tau-1}$ are fixed, the only randomness comes from the algorithm's choices of arms up to this point. Let $\alpha=(a_1,a_2,\cdots,a_\tau)$ represent a possible sequence of arms played by the algorithm for the first $\tau-1$ time steps. This event occurs with probability $Q_\alpha$, and the algorithm would accumulate a total reward of $R_\alpha$. Assuming by induction that the invariant holds on time steps up to $\tau-1$, we can state that $\sum_\alpha Q_\alpha R_\alpha \leq \frac{\tau-1}{K}$.

Now on time step $\tau$, based on its past choices of arms $\alpha$, the algorithm constructs a probability vector $p_1^\alpha, p_2^\alpha, \cdots p_K^\alpha$ representing the probabilities of playing each arm on time step $\tau$. Now we construct $b_\tau \in [K]$ independently of $\alpha$. For any fixed $b_\tau$, the expected reward by the algorithm after time step $\tau$ is given by:
\begin{align*}
E[\sum_{t=1}^{\tau}x_{a(t)}(t)]
= \sum_\alpha Q_\alpha(R_\alpha + 1 \times p_{b_\tau}^\alpha)
= \sum_\alpha Q_\alpha R_\alpha + \sum_\alpha Q_\alpha p_{b_\tau}^\alpha
\end{align*}
As we have the following:
\begin{align*}
\sum_\alpha Q_\alpha p_1^\alpha + \sum_\alpha Q_\alpha p_2^\alpha + \cdots + \sum_\alpha Q_\alpha p_K^\alpha
= \sum_\alpha Q_\alpha (p_1^\alpha + p_2^\alpha + \cdots + p_K^\alpha)
= \sum_\alpha Q_\alpha\times 1
= 1
\end{align*}
there must exist a $b_\tau \in [K]$ such that $\sum_\alpha Q_\alpha p_{b_\tau}^\alpha \leq \frac{1}{K}$.
By selecting such a $b_\tau$, we can conclude that:
\begin{align*}
E[\sum_{t=1}^{\tau}x_{a(t)}(t)]
= \sum_\alpha Q_\alpha R_\alpha + \sum_\alpha Q_\alpha p_{b_\tau}^\alpha
\leq \frac{\tau-1}{K} + \frac{1}{K} = \frac{\tau}{K}
\end{align*}
completing the inductive proof. Thus, we can lower bound the full pseudo-regret as follows:
\begin{align*}\displaystyle  R_{full}(T) = \sum_{t=1}^T \max_{i \in [K]} E_a\Big[x_i(t)\Big] - E_a\Big[\sum_{t=1}^T x_{a(t)}(t)\Big]
\geq \Big(\sum_{t=1}^T 1\Big) - \frac{T}{K} = T\frac{K-1}{K} \end{align*}

\subsection{Optimized \algoname}
\label{appendix_optimizedperiodicexp4}

In this section, we give an optimized implementation of \algoname to show that the running time bounds can be improved upon with some pre-computation. With this optimization, we have a running time of $O(K\sum_{f \in F}P_f)$ for pre-computation, and $O(K|F|)$ per time step later on. However, we note that such an implementation can potentially increase the amount of numerical error.
The reduction in running time largely comes from optimizing the following computation via the introduction of variables $S_f^\ell(t)$ and $B_f(t)$:
\begin{align*}
r_i(t) := \sum_{f \in F} \Big( b_i^{f(t),f} (t) \prod_{\ell \in f([t]) \setminus \{f(t)\}}\sumjk b_j^{\ell,f}(t) \Big)
\end{align*}
\begin{algorithm}
\caption{\algoname}
\begin{algorithmic}[1]
\Procedure{Initialization}{}
\ForEach {$f \in F$}
    \ForEach {$\ell \in f([T])$}
        \ForEach {$i \in [K]$}
            \State Initialize $b_i^{\ell,f}(1) = 1$
        \EndFor
        \State Initialize $S_f^\ell(1) = K$ \Comment we maintain $\ds S_f^\ell(t) = \sumjk b_j^{\ell,f}(t)$
    \EndFor
    \State Initialize $B_f(1) = K^{|f([T])|}$ \Comment we maintain $\ds B_f(t) = \prod_{\ell \in f([t])} S_f^\ell(t)$
\EndFor
\EndProcedure

\Procedure{Algorithm}{}
\ForEach {time step $t = 1,2,\cdots,T$}
    \ForEach {$i \in [K]$}
        \State $\ds r_i(t) := \sum_{f \in F} b_i^{f(t),f} (t) \times B_f(t) / S_f^{f(t)}(t)$ \Comment running time $O(K|F|)$
    \EndFor

    \ForEach {$i \in [K]$}
        \State $\ds p_i(t) = \frac{r_i(t)}{\sumjk r_j(t)}$ \Comment running time $O(K)$
    \EndFor

    \State Play arm $i_t \in [K]$ from the probabilities $p_1(t), p_2(t), \cdots, p_K(t)$
    \State Obtain reward $x_{i_t}(t)$

    \ForEach {$f \in F$}
        \ForEach {$\ell \in f([T])$}
            \ForEach {$i \in [K]$}
                \If {$i = i_t$ and $\ell = f(t)$}
                    \State $b_i^{\ell,f}(t+1) = b_i^{\ell,f}(t) \exp(\gammK x_i(t)/p_i(t))$ \Comment running time $O(|F|)$
                \Else
                    \State $b_i^{\ell,f}(t+1) = b_i^{\ell,f}(t)$ \Comment implemented as no-op
                \EndIf
            \EndFor
        \EndFor
        
        \State $S_f^\ell(t) := S_f^\ell(t) + b_{i_t}^{\ell,f}(t) - b_{i_t}^{\ell,f}(t-1)$ \Comment running time $O(|F|)$
        \State $B_f(t) := B_f(t-1) \times S_f^\ell(t) / S_f^\ell(t-1)$ \Comment running time $O(|F|)$
    \EndFor
\EndFor
\EndProcedure
\end{algorithmic}
\end{algorithm}

\subsection{Correctness of \algoname}
\label{appendix_proofofcorrectness}
In this section, we complete the proof of correctness of \algoname as mentioned in Section \ref{proofOfCorrectness}.
As described before, we show that our algorithm produces the same probability distribution over arms as EXP4 in every time step.
In EXP4, $\pi_{\theta,f}$ is the expert which at time $t$ recommends arm $\theta\circ f(t)$ with probability $1$ and all other arms with probability $0$. We show that \algoname is distributionally equivalent to EXP4, where $\Pi = \{\pi_{\theta, f} | f \in F, \theta \in \Theta_f\}$.
In EXP4, Each expert $\pi_{\theta,f}$ would have some weight $w_{\theta,f}(t)$ at time step $t$. At time step $t$, EXP4 plays arm $i$ with probability $p_i(t)$ represented by the following expression:
\begin{align*}
\ds p_i(t) = \frac{\sum_{f \in F, \theta \in \Theta_f, \theta\circ f(t) = i} w_{\theta,f}(t)}{\sum_{f \in F, \theta \in \Theta_f} w_{\theta,f}(t)}
\end{align*}
To show that the two algorithms are distributionally equivalent, as $p_i(t) := {r_i(t)}/{\sum_{j=1}^{K}r_j(t)}$ in \algoname, for each successive time step $t$, we only need to show the following:
\begin{align*}
\ds r_i(t) = \sum_{f \in F, \theta \in \Theta_f, \theta\circ f(t) = i} w_{\theta,f}(t) \tag*{$r_i(t)$ is defined in Algorithm \ref{alg:groupedexp4}.}
\end{align*}
We first note that for each $f \in [F]$, $\ell \in f([T])$, $i \in [K]$, from the way $b_i^{\ell,f}(t+1)$ is defined in \algoname (Algorithm \ref{alg:groupedexp4}), we have the following expression:
\begin{align*}
b_i^{\ell,f}(t+1) = \prod_{s\in f^{-1}(\ell)\cap[t]} \exp \big(\frac{\gamma}{K}\xhat_i(s)\big) = \exp \big(\frac{\gamma}{K} \sum_{s\in f^{-1}(\ell)\cap[t]} \xhat_i(s)\big)
\end{align*}
We then note that in EXP4,
\begin{align}
\ds w_{\theta,f}(t+1)& = \exp\Big(\gammK \sum_{s \in [t]} \xhat_{\theta \circ f(s)}(s)\Big) \tag*{}\\
& = \exp\Big(\gammK \sum_{\ell \in f([t])} \sum_{s \in f^{-1}(\ell)\cap[t]} \xhat_{\theta \circ f(s)}(s)\Big) \tag*{Divide up $[t]$ by label}\\
& = \exp\Big(\gammK \sum_{\ell \in f([t])} \sum_{s \in f^{-1}(\ell)\cap[t]} \xhat_{\theta (\ell)}(s)\Big) \tag*{As $\ell = f(s)$}\\
& = \prod_{\ell \in f([t])} \exp\Big(\gammK \sum_{s \in f^{-1}(\ell)\cap[t]} \xhat_{\theta (\ell)}(s)\Big) \tag*{}\\
& = \prod_{\ell \in f([t])} b_{\theta(\ell)}^{\ell,f}(t+1)
\label{eqn:proof_w}
\end{align}
where $b_{\theta(\ell)}^{\ell,f}(t)$ comes from \algoname (Algorithm \ref{alg:groupedexp4}).
We then note that:
\begin{align*}
\ds \sum_{f \in F, \theta \in \Theta_f, \theta\circ f(t) = i} w_{\theta,f}(t) = \sum_{f \in F}\Big(\sum_{\theta \in \Theta_f, \theta\circ f(t) = i} w_{\theta,f}(t)\Big)
\end{align*}

\noindent and that, (the last step is as $\Theta_f$ contains every function $\theta : f([T]) \to [K]$)
\begin{align*}
\sum_{\theta \in \Theta_f, \theta\circ f(t) = i} w_{\theta,f}(t)
& = \sum_{\theta \in \Theta_f, \theta\circ f(t) = i} \prod_{\ell \in f([t])} b_{\theta(\ell)}^{\ell,f}(t) \tag*{By (\ref{eqn:proof_w})}\\
& = \sum_{\theta \in \Theta_f, \theta\circ f(t) = i} b_{\theta \circ f(t)}^{f(t),f}(t) \prod_{\ell \in f([t])\setminus \{f(t)\}} b_{\theta(\ell)}^{\ell,f}(t) \tag*{Extract current label}\\
& = b_i^{f(t),f}(t) \sum_{\theta \in \Theta_f, \theta\circ f(t) = i} \prod_{\ell \in f([t])\setminus \{f(t)\}} b_{\theta(\ell)}^{\ell,f}(t) \\
& = b_i^{f(t),f}(t) \prod_{\ell \in f([t])\setminus\{f(t)\}} \sumjk b_j^{\ell,f}(t)
\end{align*}

Thus we have,
\begin{align*}
\ds \sum_{f \in F, \theta \in \Theta_f, \theta\circ f(t) = i} w_{\theta,f}(t) = \sum_{f \in F} \Big( b_i^{f(t),f}(t) \prod_{\ell \in f([t])\setminus\{f(t)\}} \sumjk b_j^{\ell,f}(t) \Big) = r_i(t)
\end{align*}
This shows that the expression for $r_i(t)$ we have defined in \algoname corresponds to the sum of weights of all the ``experts'' $\pi_{\theta, f}$ from EXP4 which agree to play arm $i$ on time step $t$. Thus this concludes the proof of distributional equivalence between the algorithms.

\subsection{Worst Case Regret Lower Bound on a Single Partition}
\label{appendix_partitionlowerbound}
This is based on the pseudo-regret in the setting with a single partition function (\ref{expr:single_function_periodic_regret}).

We give a proof for Theorem \ref{theorem_partitionlowerbound}.
This proof bears many similarities with the proof of a $\Omega\big(\sqrt{KT\frac{\log N}{\log K}}\big)$ lower bound for the problem of bandits with expert advice in \cite{seldin2016lower}.

We make use of a modified formulation by \cite{seldin2016lower} of a theorem originally presented in \cite{auer2002nonstochastic}.
\begin{theorem}{\cite{seldin2016lower, auer2002nonstochastic}}
\label{theorem_ckt}
Assume that the number of time steps $T \geq K/(4\ln \frac{4}{3})$. Then there exists a randomized oblivious adversary $R$, such that for algorithms $a$,
\begin{align*}
\ds \inf_a \Big( \max_{i \in [K]} E\Big[\sum_{t \in [T]} x_{i}(t)\Big] - E\Big[\sum_{t \in [T]} x_{a(t)}(t)\Big] \Big) \geq \sqrt{cKT}
\end{align*}
\end{theorem}

We note that this randomized oblivious adversary $R$ picks arms independently of the choices made by algorithm $a$. Details on the construction of this adversary $R$ are given in \cite{seldin2016lower}. $c = {(\sqrt{2}-1)}/{\sqrt{32\ln (4/3)}}$ is a constant independent of any parameter.

We now proceed with the proof of Theorem \ref{theorem_partitionlowerbound}.
For each label $\ell \in [P]$, we consider a bandit problem of length $T_\ell := |f^{-1}(\ell)|$. If we assume each $T_\ell \geq K/(4\ln \frac{4}{3})$, then by Theorem \ref{theorem_ckt}, there exists an adversary $R_\ell$ such that for any algorithm $a_\ell$ running on a bandit problem of length $T_\ell$,
\begin{align*}
\max_{i \in [K]} E\Big[\sum_{t \in [T_\ell]} x_{i}(t)\Big] - E\Big[\sum_{t \in [T_\ell]} x_{a_\ell(t)}(t)\Big] \geq \sqrt{cKT_\ell}
\end{align*}
We now construct an adversary $R$ for a bandit problem of length $T$ on partition function $f$. For each time step $t$, let $\ell := f(t)$. The adversary $R$ takes $R_\ell$'s advice to generate a randomized reward vector for time step $t$.

Now consider any algorithm $a$ for a bandit problem of length $T$ on partition function $f$, and run it against the adversary $R$. Suppose there exists a label $\ell \in [P]$ such that:
\begin{align*}
\max_{i \in [K]} E\Big[\sum_{t \in f^{-1}(\ell)} x_{i}(t)\Big] - E\Big[\sum_{t \in f^{-1}(\ell)} x_{a(t)}(t)\Big] < \sqrt{cKT_\ell}
\end{align*}
We can then consider a ``restriction'' $a_\ell$ of algorithm $a$ that plays on a bandit problem of length $T_\ell := |f^{-1}(\ell)|$. This algorithm $a_\ell$ would play exactly what algorithm $a$ would play on the $T_\ell$ time steps of label $\ell$, while simulating $a$'s plays against adversary $E$ internally on all other time steps. Therefore, against the adversary $R_\ell$, the algorithm $a_\ell$ would achieve a pseudo-regret under $\sqrt{cKT}$, which contradicts our choice of adversary $R_\ell$.

We can thus conclude that:
\begin{align*}
&\max_{\theta \in \Theta_f} E\Big[\sum_{t \in [T]} x_{\theta \circ f(t)}(t)\Big] - E\Big[\sum_{t \in [T]} x_{a(t)}(t)\Big]\\ &= \max_{(i_1,i_2,\cdots,i_P) \in [K]^P} \sum_{\ell \in [P]}\Big( E\Big[\sum_{t \in f^{-1}(\ell)} x_{i_\ell}(t)\Big] - E\Big[\sum_{t \in f^{-1}(\ell)} x_{a(t)}(t)\Big]\Big) \\
&= \sum_{\ell \in [P]} \max_{i \in [K]} E\Big[\sum_{t \in f^{-1}(\ell)} x_{i}(t)\Big] - E\Big[\sum_{t \in f^{-1}(\ell)} x_{a(t)}(t)\Big] \\
&\geq \sum_{\ell \in [P]} \sqrt{cKT_\ell}
\end{align*}

\subsection{Lower Bound on Worst Case Generalized Periodic Regret}
\label{appendix_generalizedlowerbound}
We show a lower bound on the worst-case pseudo-regret in the generalized periodic setting (\ref{expr:generalized_periodic_regret}) based on $T$, $K$, $|F|$ and $P$.
In order to make use of Theorem \ref{theorem_partitionlowerbound} later on in the proof, we first make the base assumption that $T$ is sufficiently large. Specifically, we require that $T \geq \Big(\frac{\log |F|}{\log K} + K - \frac{\log 0.5}{\log K}\Big)\times K/(4\ln \frac{4}{3})$.

Fix any integer $M \geq P$.
Suppose $K \leq P$ (more labels than arms). We split the time steps $[T]$ into $M$ equally sized sections. Now we let $F$ be the set of all partitions of $[M]=\{1,2,\cdots,M\}$ into $K$ parts. As $K \leq P$, we can define $F$ this way. Each partition in $F$ assigns one label in $[K]=\{1,2,\cdots,K\}$ to each of the $M$ sections.

\begin{lemma}
This set of partitions $F$ covers all possible ways to assign a different arm to each of the $M$ parts.
\end{lemma}

\begin{proof}
Consider any assignment $\pi: [M] \to [K]$, representing each possible assignment of arms to the $M$ parts. We have a partition function $f \in F$ that partitions $[M]$ into the pre-image sets $\pi^{-1}(1), \pi^{-1}(2), \cdots, \pi^{-1}(K)$ of $\pi$ by assigning each a separate label. The labels can then be assigned to arms accordingly to represent $\pi$.
\end{proof}

As OPT can choose an arm for each of the $M$ sections independently, we obtain a regret lower bound of $\Omega(M \times \sqrt{K \frac{T}{M}}) = \Omega(\sqrt{MKT})$ by Theorem \ref{theorem_partitionlowerbound}. Now we express $M$ in terms of $|F|$. As $|F|$ is the number of partitions of $[M]$ into $K$ parts, we have $|F| = S(M,K)$, where $S(M,K)$ refers to the Stirling numbers of the second kind. Using the upper and lower bounds for $S(M,K)$ from \cite{rennie1969stirling}, we can bound $|F|$, and thus $M$, as follows:
\begin{align*}
& \frac{1}{2} \binom{M}{K} K^{M-K} \leq |F| \leq \frac{1}{2}(K^2 + K + 2)K^{M-K-1} - 1 \\
&\implies \frac{1}{2}K^{M-K} \leq |F| \leq \frac{1}{2}(K^2 + K^2 + 2K^2)K^{M-K-1} = 2K^{M-K+1} \\
&\implies \frac{\log |F|}{\log K} + (K-1) - \frac{\log 2}{\log K} \leq M \leq \frac{\log |F|}{\log K} + K - \frac{\log 0.5}{\log K}
\end{align*}

This expression validates our use of Theorem \ref{theorem_partitionlowerbound}, as this with our assumption implies that $T/M \geq K/(4\ln \frac{4}{3})$, and thus for any $f \in F$, there are at least $K/(4\ln \frac{4}{3})$ time steps associated with each label.
From the bounds on $M$, we have that $M = \Theta (\frac{\log |F|}{\log K} + K)$. This thus gives us a lower bound of $\Omega\Big(\sqrt{KT(\frac{\log|F|}{\log K} + K)}\Big)$.

Now, we can add a single extra partition function to $F$, that partitions the time steps $[T]$ into $P$ equal parts. With the use of Theorem \ref{theorem_partitionlowerbound} again (note that $T/P \geq T/M \geq K/4(\ln \frac{4}{3})$), the presence of this partition function in $F$ makes our lower bound no smaller than $\Omega(\sqrt{PKT})$. Therefore, we obtain a final lower bound of:
\begin{align*}
&\Omega\Big(\sqrt{PKT} + \sqrt{KT (\frac{\log |F|}{\log K} + K)}\Big) \\
=& \Omega\Big(\sqrt{PKT} + \sqrt{KT \frac{\log |F|}{\log K}}\Big) \text{ (as $K \leq P$)}.
\end{align*}

If $P < K$ instead, a simple lower bound can be obtained by using only $P$ out of the $K$ arms, so we obtain a problem with $P$ arms and maximum partition size $P$. This gives us a lower bound of $\Omega\Big(\sqrt{PKT} + \sqrt{PT \frac{\log |F|}{\log P}}\Big)$.
We can then merge these two lower bounds into a single expression
$\Omega\Big(\sqrt{PKT} + \sqrt{\min(P,K)T \frac{\log |F|}{\log \min(P,K)}}\Big)$.

%% file: appendix_experiments.tex
\subsection{Learning Patterns}
\label{subsection:learning_patterns}

Figure \ref{figure:disc_total_probability_plots} illustrates how the average probabilities of the three networks vary with time in the discrete setup. The conclusions we can draw are the same as with the continuous setup. We only include this figure here for completeness.

\begin{figure*}[!htb]
    \centering
    \begin{subfigure}[t]{0.8\textwidth}
        \centering
        \includegraphics [scale=0.7]
{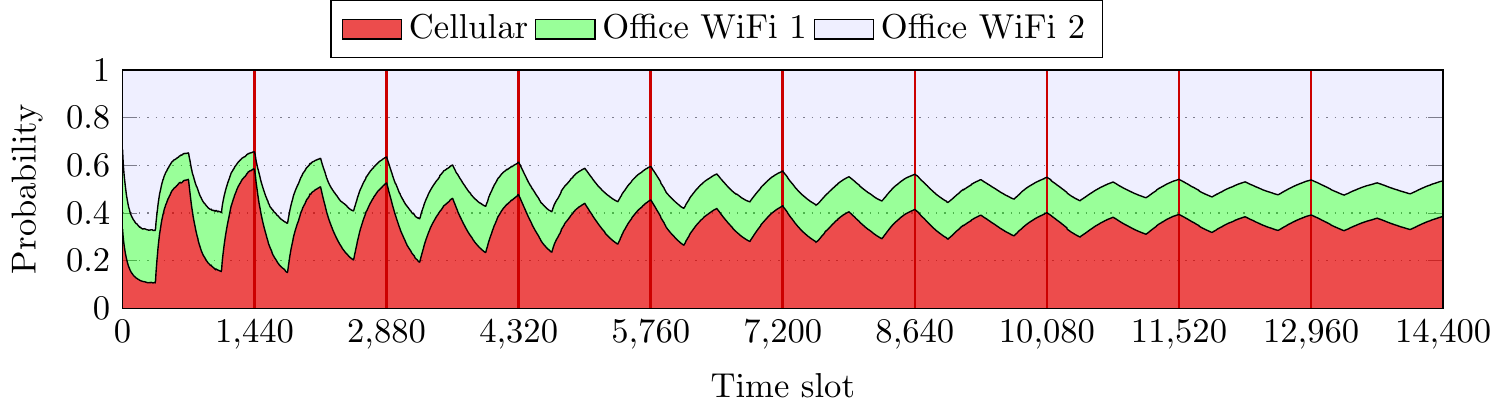}
        \caption{EXP3: Combined probabilities of each network over first 10 iterations.}
        \label{subfigure:disc_total_probability_EXP3}
    \end{subfigure}%
    ~~
    \begin{subfigure}[t]{0.2\textwidth}
        \centering
        \includegraphics [scale=0.7]
{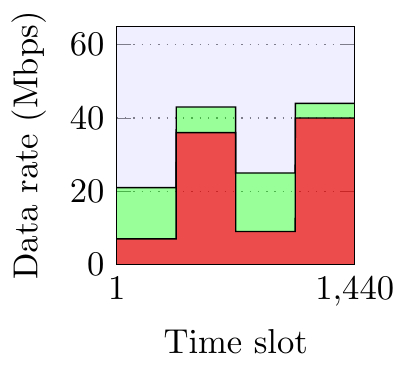}
        \caption{Bandwidth ratio}
        \label{subfigure:disc_bandwidth_ratios}
    \end{subfigure}
    ~
    \begin{subfigure}[t]{\textwidth}
        \centering
        \includegraphics [scale=0.7]
{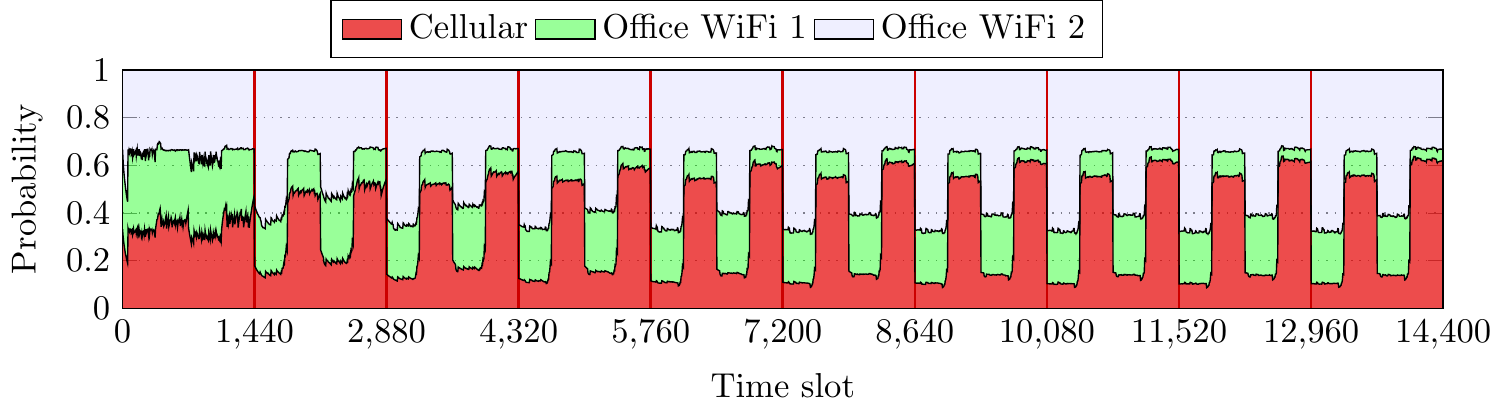}
        \caption{\algoname: Combined probabilities of each network over first 10 iterations.}
        \label{subfigure:disc_total_probability_periodicEXP3}
    \end{subfigure}
    \caption{Area chart showing the time variation of combined probabilities in the discrete setup. Figure \ref{subfigure:disc_bandwidth_ratios} shows the actual ratio of the bandwidths of the three networks within any one iteration.}
    \label{figure:disc_total_probability_plots}
\end{figure*}

\subsection{Noisy Settings}
\label{subsection:noisy_settings}
    \input{simulation/noisySettings}
    
\subsection{Comparison of Different Periods}
\label{subsection:period_comparison}
    \input{simulation/periodComparison}

\subsection{Performance Comparison of Algorithms.}
\label{subsection:performance_comparison}
\input{simulation/appendix_performanceComparison}

\subsection{Mobility of Users}
    \input{simulation/appendix_mobility}
    \label{subsection:appendix_mobility_setting}

%% file: simulation/noisySettings.tex
In this experiment, we apply a 10\% Gaussian noise to each of the networks' data rates in the instances given in Figure \ref{figure:learning_pattern_data_rates}. These noisy instances are illustrated in \ref{fig:noisy_instances}. Figure \ref{fig:score_noisy_instances} shows that \algoname is largely unaffected by noise in the data, giving similar results with and without noise.

\begin{figure*}[!htb]
    \centering
    \begin{subfigure}[t]{0.5\textwidth}
        \centering
        \includegraphics [scale=0.7]
{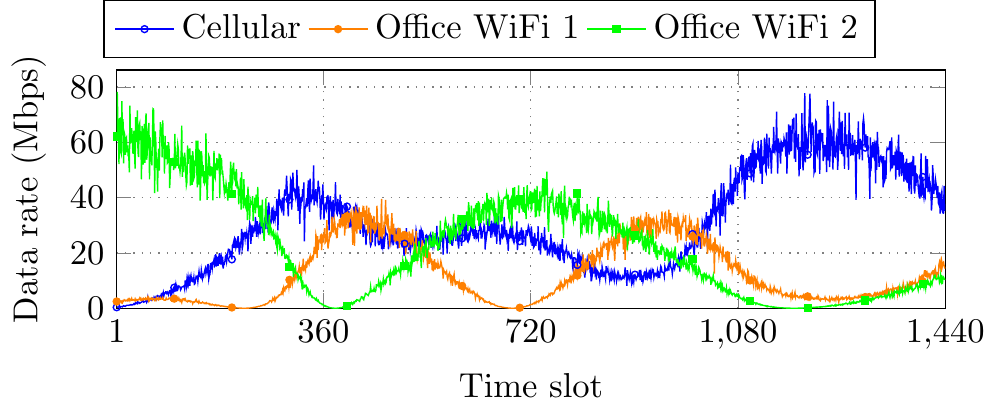}
        \caption{Continuous, 10\% gaussian noise.}
        \label{subfigure:cont_noise}
    \end{subfigure}%
    \begin{subfigure}[t]{0.5\textwidth}
        \centering
        \includegraphics [scale=0.7]
{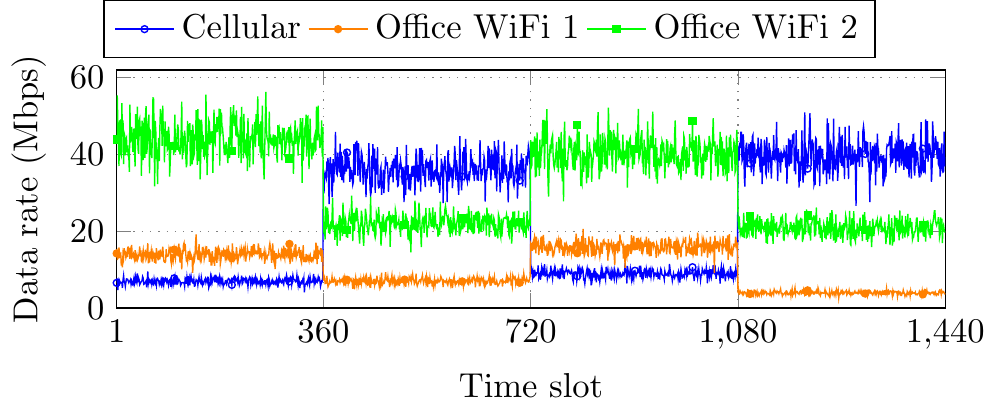}
        \caption{Discrete, 10\% gaussian noise.}
        \label{subfigure:disc_noise}
    \end{subfigure}
    \caption{Variation of data rates with within one iteration in each setup with noise. The original instances without noise are given in Figure \ref{figure:learning_pattern_data_rates}.}
    \label{fig:noisy_instances}
\end{figure*}

\begin{figure*}[!htb]
    \centering
    \begin{subfigure}[t]{0.5\textwidth}
        \centering
        \includegraphics [scale=0.7]
{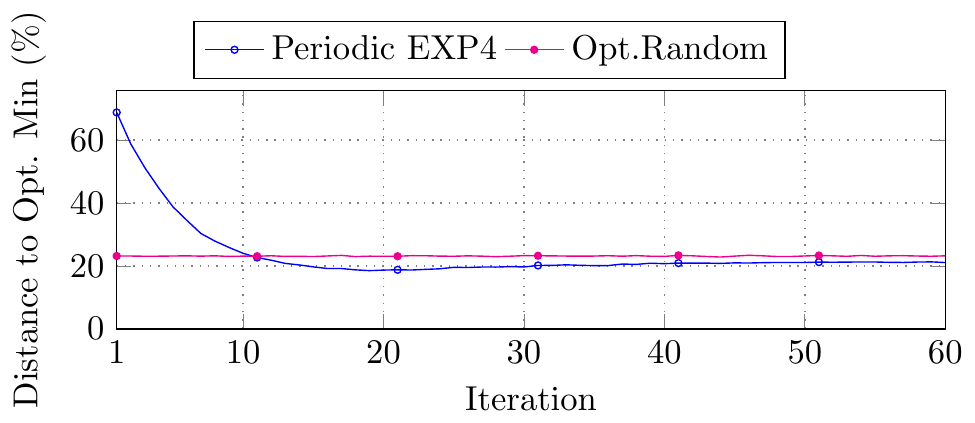}
        \caption{Continuous, no noise.}
        \label{subfigure:score_cont_no_noise}
    \end{subfigure}%
    \begin{subfigure}[t]{0.5\textwidth}
        \centering
        \includegraphics [scale=0.7]
{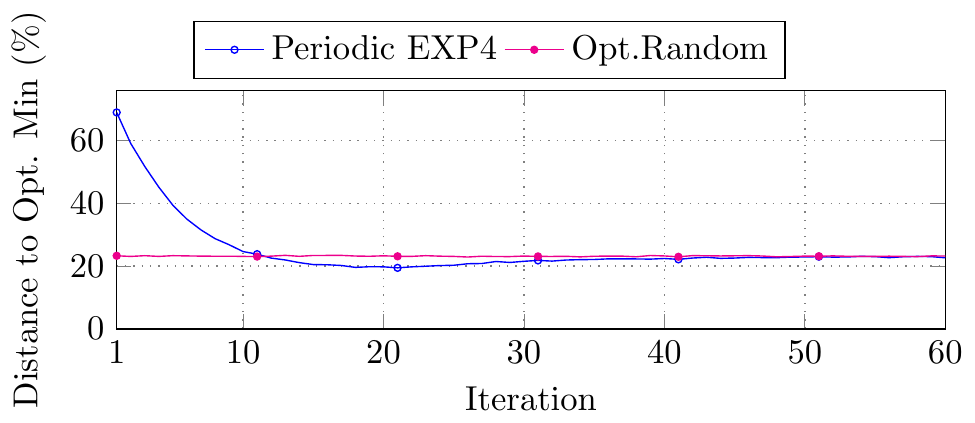}
        \caption{Continuous, 10\% gaussian noise.}
        \label{subfigure:score_cont_noise}
    \end{subfigure}
    \begin{subfigure}[t]{0.5\textwidth}
        \centering
        \includegraphics [scale=0.7]
{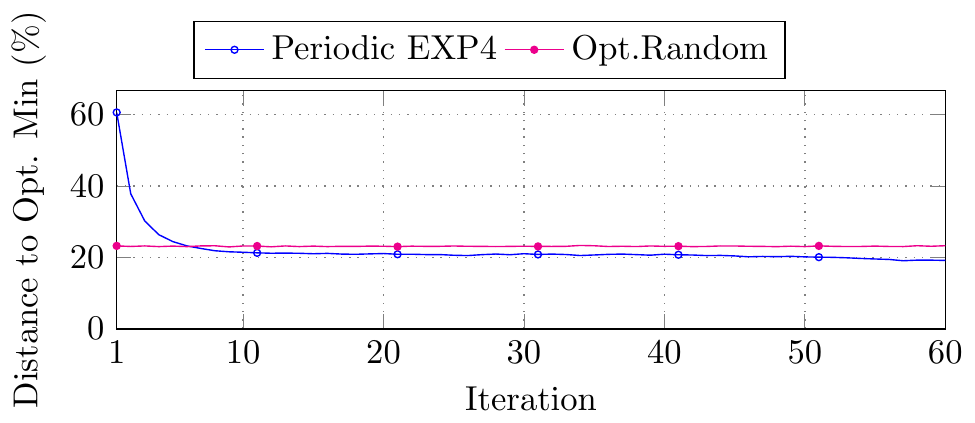}
        \caption{Discrete, no noise.}
        \label{subfigure:score_disc_no_noise}
    \end{subfigure}%
    \begin{subfigure}[t]{0.5\textwidth}
        \centering
        \includegraphics [scale=0.7]
{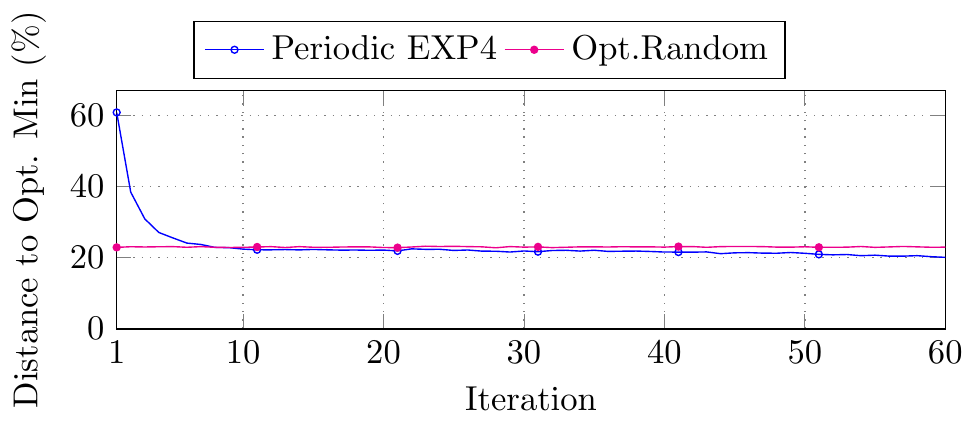}
        \caption{Discrete, 10\% gaussian noise.}
        \label{subfigure:score_disc_noise}
    \end{subfigure}
    \caption{Distance to Optimal minimum plotted over each iteration for each setup.}
    \label{fig:score_noisy_instances}
\end{figure*}

%% file: simulation/periodComparison.tex
It is useful to run \algoname with multiple periods, especially in cases where the best period is not obvious. The is the case if, for example, the bit rates (or more generally, rewards of pulling arms) vary continuously with time.

Figures \ref{fig:compare_periods} and \ref{fig:compare_periods_hard} compares the performance of \algoname when run with different period sets. The period sets chosen are $\{1\}$, $\{4\}$, $\{1,2,\cdots,15\}$, $\{1,2,\cdots,24\}$ and $\{1,2,\cdots,45\}$. Running \algoname with period set $\{1\}$ is equivalent to running EXP3, and running with period set $\{4\}$ is equivalent to running four separate instance of EXP3, switching between the four instances every quarter of an iteration. Period set $\{4\}$ is used for comparison as the discrete setting we use has a known ``best period'' of 4. It is important to note, however, that the ``best period'' is often not known prior to the experiment in practice. We also include the Optimal Random player as a point of comparison.

\begin{figure*}[!htb]
    \centering
    \begin{subfigure}[t]{0.5\textwidth}
        \centering
        \includegraphics [scale=0.7]
{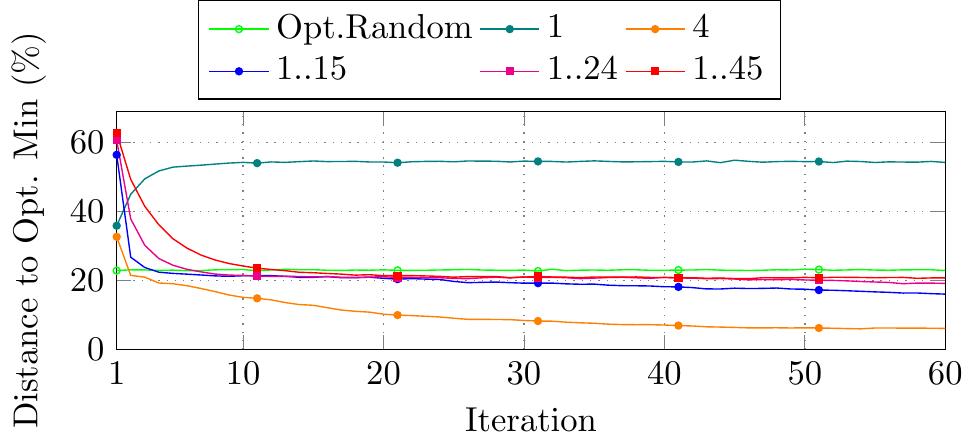}
        \caption{Discrete setting from Figure \ref{subfigure:discrete_data_rate}.}
        \label{subfigure:compare_periods_disc}
    \end{subfigure}%
    \begin{subfigure}[t]{0.5\textwidth}
        \centering
        \includegraphics [scale=0.7]
{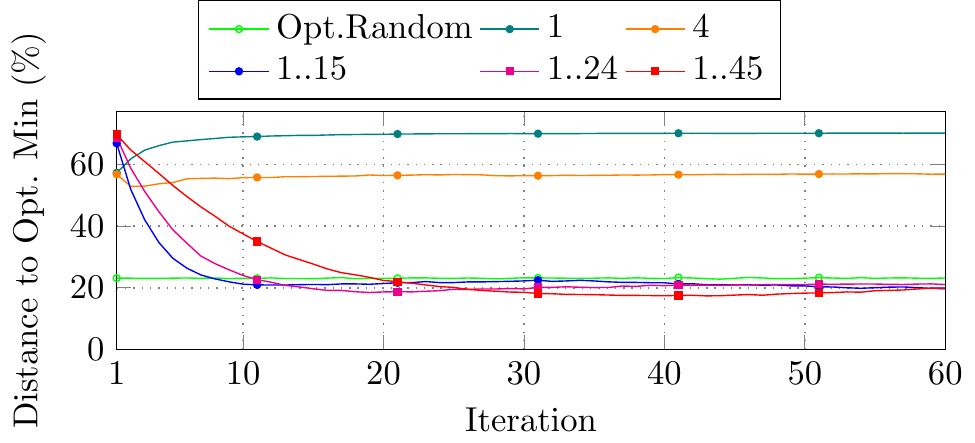}
        \caption{Continuous setting from Figure \ref{subfigure:continuous_data_rate}.}
        \label{subfigure:compare_periods_cont}
    \end{subfigure}
    \caption{Comparison of different period sets on the discrete and continuous settings.}
    \label{fig:compare_periods}
\end{figure*}

From the results on the discrete setting (Figure \ref{subfigure:compare_periods_disc}), we can see period set $\{4\}$ immediately taking the lead, and having significantly better performance than other period sets. This is to be expected, as in this case, the algorithm does not need to figure out what the ``right period'' is. Period set $\{4\}$ however performs poorly on the continuous settings.

\begin{figure*}[!htb]
    \centering
    \begin{subfigure}[t]{\textwidth}
        \centering
        \includegraphics [scale=0.7]
{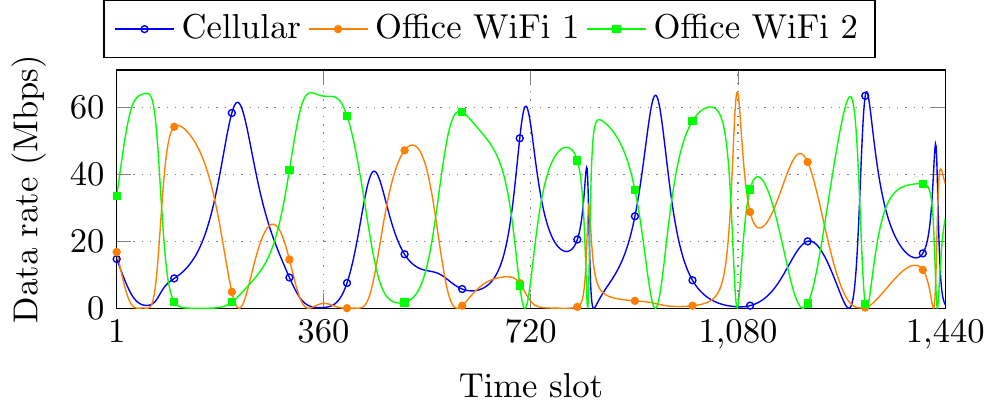}
        \caption{``Harder'' continuous setting with more bandwidth fluctuations.}
        \label{subfigure:datarates_cont_hard}
    \end{subfigure}
    \begin{subfigure}[t]{\textwidth}
        \centering
        \includegraphics [scale=0.7]
{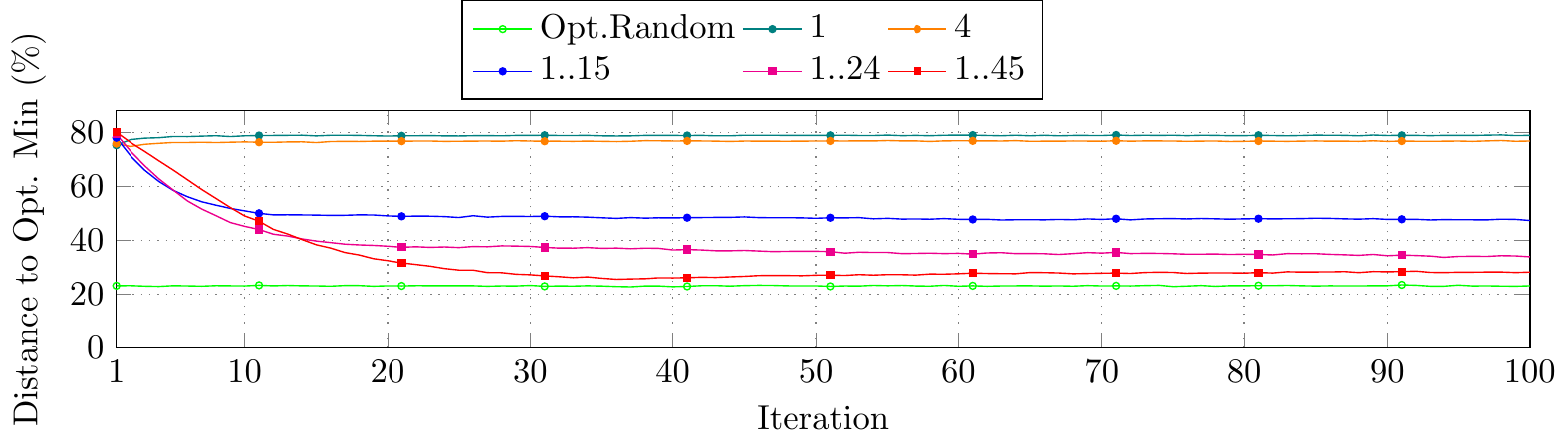}
        \caption{Comparison on the ``harder'' continuous setting.}
        \label{subfigure:compare_periods_cont_hard}
    \end{subfigure}
    \caption{Comparison of different period sets on a continuous setting with more bandwidth fluctuations. This is run for 100 iterations, a little over three months.}
    \label{fig:compare_periods_hard}
\end{figure*}

To better illustrate the utility of larger period sets, we run another experiment on a more complex continuous setting (Figure \ref{subfigure:datarates_cont_hard}), with bandwidths that fluctuate more than the setting in Figure \ref{subfigure:continuous_data_rate}. From this experiment, we can see that while larger period sets learn more slowly, they can converge to better results after a sufficiently long period of time. However, on simpler instances like in Figure \ref{subfigure:compare_periods_cont}/\ref{subfigure:continuous_data_rate}, the advantage from having a larger period set is less significant.

In summary, larger period sets have greater utility when more fluctuations in bandwidths (or rewards) are expected, as it allows the algorithm to more closely match the target pattern. However, a larger period set means the algorithm will take a longer time to learn. This is in line with the usual trade-offs between flexibility and efficiency in machine learning problems.

%% file: simulation/appendix_performanceComparison.tex
Figure \ref{figure:algo_comparison_cumulativeGain} provides the complete distribution of the cumulative gain each device observes when using \algoname, EXP3 or Optimal Random. Table \ref{table:performanceComparison_cumulativeGain} lists the median and standard deviation for each algorithm. The median and standard deviation values are computed over 20 runs with 20 devices each, for a total of 400 data points. All three have very similar median cumulative gains. We do note however, that the median cumulative gain is unlikely to indicate much, as it is likely to be close to the average cumulative gain, which as mentioned before, is not very useful as a metric as it is maximized whenever there is at least one device in each network.
On the other hand, we can see that \algoname has a lower variance than EXP3, which suggests that \algoname divides the bandwidths more evenly between the devices.

\input{table/performanceComparison_cumulativeGain}

\begin{figure}[!htb]
\begin{center}
\includegraphics [scale=0.9]
{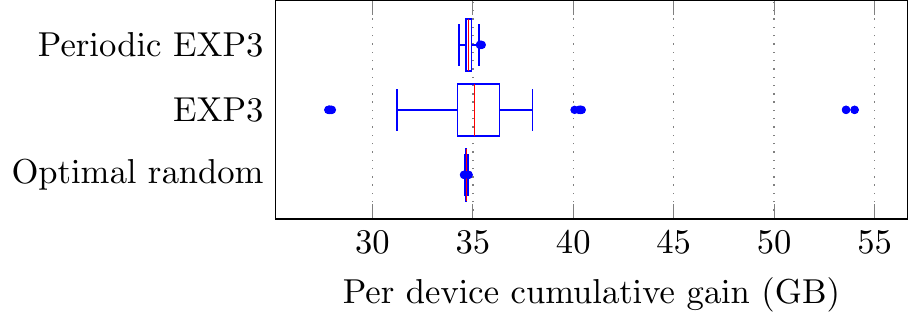}
\end{center}
\caption{Cumulative gain (GB) of a single device when using each of the algorithms, taking into account all devices across all runs --- the box represents the interquartile range (middle 50\% cumulative gains), the vertical line in the box denotes the median, the whiskers show the range of the remaining cumulative gains, excluding outliers shown as dots.}
\label{figure:algo_comparison_cumulativeGain}
\end{figure}

%% file: table/performanceComparison_cumulativeGain.tex
\begin{table}[!htb]
\small
\centering
\caption{Median and standard deviation of cumulative download (GB) per device, when using each of the three algorithms, namely \algoname, EXP3 and Optimal Random.}
\label{table:performanceComparison_cumulativeGain}
\newcolumntype{L}{>{\centering\arraybackslash}p{2.5cm}}
\begin{tabular}{l c c c }
\hline 
& \textbf{Periodic EXP4}     & \textbf{EXP3}     & \textbf{Optimal Random}     \\ \hline
\multicolumn{1}{l}{\textbf{Median (GB)}} & 34.78 & 35.08 & 34.67 \\ 
\multicolumn{1}{l}{\textbf{Standard deviation (GB)}} & 0.20  & 2.72  & 0.04  \\ \hline
\end{tabular}
\end{table}

%% file: simulation/appendix_mobility.tex
We consider a similar setting to Figure \ref{figure:heterogeneousNetworks}, with 20 mobile users (each with a mobile device) and 9 networks. Devices have access to a different set of networks over time depending on their locations. We divide an iteration into 6 phases, based on the mobility of users, as listed in Table \ref{table:mobility_setting}. In phase 1, all the mobile users are at home (same building with common networks); 10 devices have access to networks 1, 2 and 3 while the remaining have access to networks 1 and 2. We assume 5 users are always at home (whose devices have access to networks 1, 2 and 3 --- hence, they might need to switch networks when the others go out). The rest spend 13 hours at home. In phase 2, they travel (one hour) together to office, during which they have access to networks 4 and 5. In phase 3, they have access to networks 6, 7 and 8 at the office. After 3 hours of work, in phase 4, 10 users go out for a 1-hour lunch during which they have access to networks 8 and 9. After lunch, they spend another 5 hours at the office (phase 5), before travelling home for an hour in phase 6. 

\input{table/mobility_setting}

Vanilla \algoname is oblivious to networks possibly becoming unavailable, and we give it a gain of zero whenever it decides to select an inaccessible network, in the hope that it also learns the pattern of network availability. We compare it against an optimized version of \algoname which selects only from the set of currently available networks. Figure \ref{figure:mobility_distanceToNashEquilibrium} shows that, as expected, the optimized version initially yields a better performance. However, after the Vanilla \algoname algorithm learns the pattern, they both perform equally well. 
The intuition for a slightly better performance compared to Optimal Random 
is that when using \algoname, many of the devices eventually converge to selecting a single network with high probability, and are thus less likely to make a random \enquote{bad} selection. On the other hand, Optimal Random has every device running the exact same (though distributionally optimal) probability distribution, leading to the occasional random event where too many devices select the same network simultaneously.

\begin{figure}[!htb]
\begin{center}
\includegraphics 
{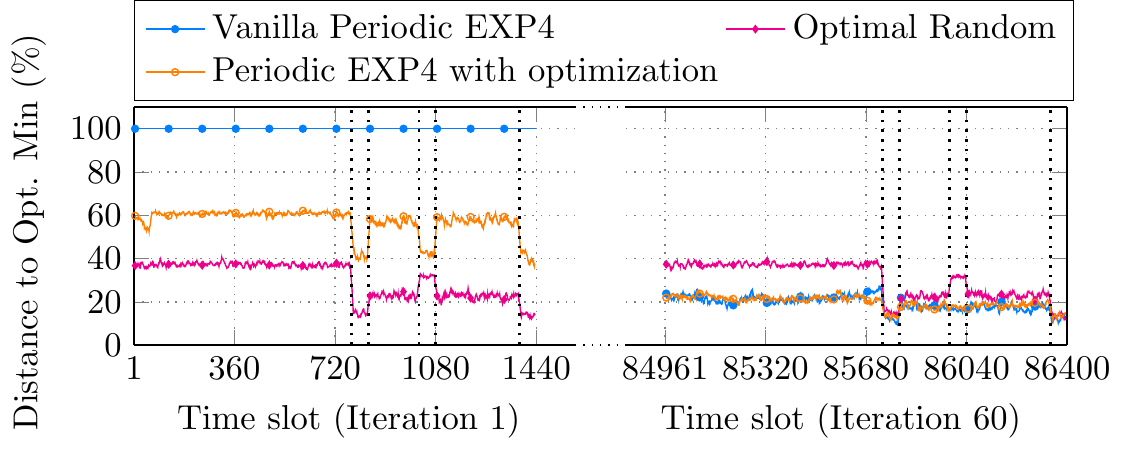}
\end{center}
\caption{
Distance to Optimal minimum in the first and last repetitions of the mobility setting. Vertical lines indicate the time of any change in the environment.}
\label{figure:mobility_distanceToNashEquilibrium}
\end{figure}

%% file: table/mobility_setting.tex
\begin{table}[!htb]
\small
\centering
\caption{Phases during one iteration, the time slots delimiting each phase (relative to the first time slot of the iteration), and the list of networks available to each device during every phase.}
\label{table:mobility_setting}
\newcolumntype{L}{>{\centering\arraybackslash}p{2.5cm}}
\begin{tabular}{l c c c c}
\hline
\textbf{Phase(s)}                                                                                                  & \textbf{1}      & \textbf{2 and 6}     & \textbf{3 and 5}       & \textbf{4}      \\ \hline
\multirow{2}{*}{\textbf{\begin{tabular}[c]{@{}l@{}}Time slots delimiting \\ each phase\end{tabular}}} & $1 \cdots 780$ & $781 \cdots 840$ & $841 \cdots 1020$ & $1021 \cdots 1080$\\ 
 & $1381 \cdots 1440$ & $1081 \cdots 1380$ &  \\\hline
\multirow{2}{*}{\textbf{\begin{tabular}[c]{@{}l@{}}Device(s): their list\\ of available networks\end{tabular}}} & 1 - 10: 1, 2, 3 & 1 - 5: 1, 2, 3 & 1 - 5: 1, 2, 3   & 1 - 5: 1, 2, 3  \\
                                                                                                                & 11 - 20: 1, 2   & 6 - 20: 4, 5   & 6 -  20: 6, 7, 8 & 6 - 11: 6, 7, 8 \\
                                                                                                                &                 &                &                  & 11- 20: 8, 9 \\ \hline
\end{tabular}
\end{table}